\pdfoutput=1
\documentclass[12pt,fleqn]{article}
\usepackage[utf8]{inputenc} 
\usepackage[T1]{fontenc}    
\usepackage{hyperref}       
\usepackage{url}            
\usepackage{booktabs}       
\usepackage{amsfonts}       
\usepackage{nicefrac}       
\usepackage{microtype}      
\usepackage{xcolor}         
\usepackage[labelfont=bf]{caption}

\usepackage{natbib}

\usepackage{amsmath}
\usepackage{amsthm}
\usepackage{amssymb}
\usepackage{graphicx}
\usepackage{dsfont}
\usepackage{multirow}
\usepackage{subfigure}

\newtheorem{theorem}{Theorem}

\newtheorem{assumption}[theorem]{Assumption}

\newcommand{\E}{{\mathbb{E}\,}}
\newcommand{\Var}{\mathrm{Var}}

\author{%
  \small Aiyou Chen \\
  \small Google \\
  \and
  \small Nick Doudchenko \\
  \small Google \\
  \and
  \small Shunhua Jiang \\
  \small Columbia \\
  \and
  \small Cliff Stein \\
  \small Google \\
  \and
  \small Bicheng Ying \\
  \small Google \\
}

\title{Supergeo Design: Generalized Matching for Geographic Experiments}

\makeatletter
\newcommand{\printfnsymbol}[1]{%
  \textsuperscript{\@fnsymbol{#1}}%
}
\makeatother

\begin{document}
\maketitle

\begin{abstract}
\noindent We propose a generalization of the standard matched pairs design 
in which experimental units (often geographic regions or \emph{geos}) may be combined into larger units/regions called ``supergeos'' in order to improve the average matching quality. 
Unlike optimal matched pairs design which can be found in polynomial time \citep{lu2011optimal}, this generalized matching problem is NP-hard. We formulate it as a mixed-integer program (MIP) and show that experimental design obtained by solving this MIP can often provide a significant improvement over the standard design regardless of whether the treatment effects are homogeneous or heterogeneous. Furthermore, we present the conditions under which trimming techniques that often improve performance in the case of homogeneous effects \citep{chen2022robust}, may lead to biased estimates and show that the proposed design does not introduce such bias. We use empirical studies based on real-world advertising data to illustrate these findings.
\end{abstract}

\section{Introduction}\label{intro}
With online advertising revenue in the US amounting to almost 200 billion dollars in 2021 \citep{iab_report_2021}, it is both theoretically and practically important to understand advertising effectiveness. For instance, researchers may want to know how much additional sales revenue does an additional dollar spend on advertising generate? 
Ideally, such questions would be addressed via
the ``gold standard'' of causal inference---randomized experiments or A/B tests as they are often called in applied settings.  

When many distinct units are available to experiment on and the estimand of interest is the \emph{average treatment effect} (ATE), simple randomized experiments may provide precise and intuitive results. Unfortunately, such experiments are not always available or, even when they are, may not be adequate for the question at hand. For instance, if interference between the experimental units---a situation in which the treatment status of one unit affects the outcomes of some other units---is not negligible, the researchers may decide to combine units into larger groups (or clusters) if they have reasons to believe that interference at the group level is less of a concern. In other cases, institutional constraints or privacy concerns may prevent the researchers from assigning the treatment at a more granular level. Neither of these concerns are foreign to the studies of advertising effectiveness \citep{coey2016people,vaver2011measuring}.   
For example, product information learned from an ad by one person can be easily propagated to another person who has not seen the ad introducing interference. Different variations of television advertising are often assigned at the designated market area (DMA) level limiting the number of experimental units to only 210 DMAs existing in the US.\footnote{\url{https://markets.nielsen.com/us/en/contact-us/intl-campaigns/dma-maps/}} Moreover, targeted advertising is an important subject of current academic and public discourse with an increasing number of studies conducted at geographic---instead of more granular---levels \citep[Google Ads geo targets,\footnote{\url{https://developers.google.com/google-ads/api/reference/data/geotargets}}][]{rolnick2019randomized}. 
These considerations may not only limit the sample size, but also lead to experimental units that vary substantially along observed---and potentially unobserved---characteristics, thereby preventing the researchers from relying on large-sample properties alone and forcing them to make additional assumptions.

We are primarily motivated by evaluating advertising effectiveness using geographic experiments. However, these issues---relatively small sample size and heterogeneity of experimental units---may come up in other settings. Regardless of the setting, there are two general ways to improve the statistical properties of experimental results---by changing the allocation of units to treatment and control (\emph{experimental design}) or by changing the approach to estimating the quantity of interest (\emph{post-experimental analysis}). One common approach to experimental design is to find disjoint pairs of comparable units and randomize the treatment within each pair---the so called \emph{matched pairs design} \citep{imbens2015causal}. This experimental design has several advantages:  
(i) it allows---at least to some degree---balancing the control and treatment groups with respect to the observed unit-level covariates, (ii) the treatment effects can be estimated separately for each pair which may be particularly useful when the effects are heterogeneous, (iii) randomizing the treatment within each pair allows rigorous statistical inference (e.g.~permutation-based tests). Optimal matched pairs can be designed by solving a minimum-weight matching problem  
\citep[see, for example, Chapter~12 of][]{rosenbaum2020design} which can be done in polynomial time \citep{edmonds1965maximum}.\footnote{See \citet{stuart2010matching}, \citet{rosenbaum2020design} and \citet{pashley2021insights} for additional references to various general-purpose matching methods.}

Unfortunately, it is not always possible to construct pairs of units that are sufficiently similar. It is not uncommon in applications to have pairs of units that are so different from each other that it is better to exclude some of them from the experiment completely in order to make the estimates more precise. This can be done either at the design phase \citep{chen2021trimmed} or during post analysis \citep{chen2022robust}. However, such ``trimming'' reduces the sample size and, perhaps more importantly, makes the resulting sample potentially different from the original population in terms of the average treatment effect if the effects are heterogeneous across units. Moreover, excluding pairs of units from the experiment may require stopping advertising in those units for the duration of the experiment. This may lead to a drop in sales revenue which otherwise could have been avoided.\footnote{We discuss the difference between ``trimming'' in the design and post-experimental analysis phases in a bit more detail in Section~\ref{sec:eval}.}

\paragraph{Supergeo design.} To combat the issue of poor matches without sacrificing the data, we introduce a generalized matching problem in which each pair is a pair of ``superunits'' with each superunit being a sum of several original units. Conceptually, this extends the classical matched pairs to matched superunit-pairs.\footnote{The notion of ``superunit'' may not apply to all settings as the covariates and/or response variables may not be additive. For example, in medical science, it may not be meaningful to pair a group of two patients to another group of three patients. However, other aggregation approaches---besides the simple summation---may be considered depending on the application.} Since our primary motivation for this paper is geographic experiments where each unit is a DMA, in the remainder of the paper we use \emph{geo} and \emph{supergeo} as synonyms for unit and superunit respectively. We use the term \emph{supergeo design} to refer to the proposed methodology.

Finding the optimal supergeo design is closely related to the minimum-weight matching problem over the hyper-graph \citep{keevash2014geometric}. What differentiates our paper from previous algorithmic work is the fact that achieving good matches is not the end goal, but rather an approach to improving the statistical properties of the experimental results. Subsequently, we evaluate the proposed methodology as an experimental design tool.

We show that supergeo design can improve the precision of the average treatment effect estimates when compared to standard matched pairs design. 
Moreover, it can substantially outperform the alternatives that require excluding some of the geos from the experiment when the treatment effects are heterogeneous since those alternative designs may lead to a subpopulation with a different average treatment effect compared to the original sample. 

Although further generalizations \citep[for example, in the spirit of the synthetic-control literature,][]{abadie2010synthetic} are possible, we currently assume that when the units are aggregated, the covariates used for matching are added up. For instance, this is appropriate when the covariate of interest is the total sales within a geo and we are willing to assume that as long as the time series of past sales data are matched well, the units should respond to treatment similarly.

\paragraph{Contributions.} The main contributions of the paper are as follows:
\begin{itemize}
\item We propose a novel experimental design approach---supergeo design---that improves the balancing properties of matched pairs design. 
\item We show that supergeo design can be particularly effective when the underlying treatment effects are heterogeneous since it allows achieving good treatment-control balance without excluding geos which would normally lead to poorly matched pairs.
\item While this generalized matching problem is NP-hard, we propose an approach utilizing mixed-integer programming (MIP) that  solves the problem in practically important settings with hundreds of experimental units (representing geographic regions).
\item We validate the effectiveness of the proposed design using real advertising data.
\end{itemize}


\paragraph{Structure of the paper.} We start by introducing a theoretical framework that allows us to discuss the statistical properties of alternative experimental design and analysis procedures. We then describe the methodology used for solving the supergeo matching problem. The empirical results section of the paper compares supergeo design to the matched pairs design. We also compare supergeo design to \emph{trimmed match estimator} \citep{chen2022robust} which is an analysis---rather than design---procedure and discuss the way in which supergeo design and trimmed match can be combined.
\section{Model}\label{sec:model}
Let $\mathcal{G}$ denote the set of experimental units (geos) and let $N := |\mathcal{G}|$.\footnote{If $\mathcal{G}$ is the set of DMAs in the US, $N=210$.} For each $g \in \mathcal{G}$, $S_g$ is a variable controlled by the researcher and $R_g$ is the response variable presumably affected by $S_g$. For advertising experiments, $S_g$ and $R_g$ correspond to the advertising spend and returns respectively. Following the Neyman–Rubin causal framework \citep{imbens2015causal}, we use $(R_g^{(t)}, S_g^{(t)})$ and $(R_g^{(c)}, S_g^{(c)})$ to denote the potential outcomes under the treatment condition and the control condition, respectively. 

The following ratio measures the unit-level causal effect of the treatment:
\begin{align*}
    \theta_g := \frac{R_g^{(t)} - R_g^{(c)}}{S_g^{(t)} - S_g^{(c)}},  \;\;\;\forall g \in \mathcal{G}.
\end{align*}
We are mainly interested in estimating the following ratio:
\begin{align*}
    \theta := \frac{\sum_g R_g^{(t)} - \sum_g R_g^{(c)}}{\sum_g S_g^{(t)} - \sum_g S_g^{(c)}},
\end{align*}
which measures the global causal response rate with respect to the controlled variable. In advertising experiments the quantity $\theta$ is often called the \emph{incremental return on ad spend} (iROAS)---the ratio of the additional return to the additional advertising spend \citep{chen2022robust}. If the response is measured in the same units (e.g.~USD) as the ad spend, $\theta>1$ indicates that it makes sense to expend an additional dollar on advertising.

\paragraph{Design and estimators.}
In practice we never observe both $(R_g^{(t)}, S_g^{(t)})$ and $(R_g^{(c)}, S_g^{(c)})$ since we can only assign a unit to either treatment or control. Let us use $\mathcal{T}$ and $\mathcal{C} \subseteq \mathcal{G}$ to denote the treatment and control groups in the experiment.
The standard way to estimate $\theta$ is to use an \emph{empirical estimator} that computes the ratio of differences between the treatment and control groups:
\begin{align}\label{eq:theta_hat}
\hat{\theta} := \frac{\sum_{g \in \mathcal{T}} R_g - \sum_{g \in \mathcal{C}} R_{g}}{\sum_{g \in \mathcal{T}} S_g - \sum_{g \in \mathcal{C}} S_{g}}.
\end{align}
The researcher might want to carefully choose the treatment group, $\mathcal{T}$, and the control group, $\mathcal{C}$, so that $\hat{\theta}$ is a good---in the statistical sense---estimate of $\theta$.\footnote{In the remainder of the paper we often refer to this experimental design stage as the \emph{pretest} phase as opposed to the \textit{test} phase which refers to the stage when the treatment is applied and the experimental data are collected and analyzed.}
For example, in the classical matched pairs design, geos are allocated into pairs, $\{(g_{k,+}, g_{k,-})\}_{k=1}^{N/2}$. Then, one of the geos in each pair is randomly assigned to the treatment group while the other one is assigned to the control group.

An alternative estimator that is often more precise than the empirical estimator is the \emph{trimmed match estimator} from \citet{chen2022robust}. This estimation procedure amounts to selecting the subsets $\mathcal{T}' \subseteq \mathcal{T}$ and $\mathcal{C}' \subseteq \mathcal{C}$ in order to obtain the most precise estimate of the form:\footnote{See \citet{chen2022robust} for the specifics of how the subsets $\mathcal{T}'$ and $\mathcal{C}'$ are selected.}

\begin{equation}\label{eq:theta_hat_trim}
\hat{\theta}^{\mathrm{trim}} := \frac{\sum_{g \in \mathcal{T}'} R_g - \sum_{g \in \mathcal{C}'} R_{g}}{\sum_{g \in \mathcal{T}'} S_g - \sum_{g \in \mathcal{C}'} S_{g}}.
\end{equation}

\paragraph{Model assumptions.}
We follow \cite{chen2022robust} and make a number of assumptions that facilitate a formal analysis of alternative experimental design and analysis procedures.

\begin{assumption}[Linear response model]\label{ass:linear}
For each $g \in \mathcal{G}$ the \emph{uninfluenced response},\footnote{Corresponding to the response under zero advertising spend.} $Z_g$, does not depend on treatment assignment, and $R_g$ depends on $S_g$ linearly:
\begin{align*}
R_g = Z_g + \theta_g \cdot S_g.
\end{align*}
\end{assumption}

\noindent The next assumption stipulates that the total difference in spend between the treatment and control units is unaffected by the treatment assignment. It captures the real-life scenario in which the researcher has a budgetary constraint when conducting the experiment. It also simplifies the calculation of the variance of $\hat{\theta}$.
\begin{assumption}[Fixed budget]\label{ass:fixed_budget}
There is a fixed experimental budget $B$ so that
\[
\sum_{g \in \mathcal{T}} S_g - \sum_{g \in \mathcal{C}} S_{g} = B.
\]
\end{assumption}



\paragraph{Supergeo design.}

The supergeo experimental design proposed in this paper consists of $K$ \emph{supergeo pairs} $\{(G_{k,+}, G_{k,-})\}_{k=1}^K$ where $G_{k,+}, G_{k,-} \subset \mathcal{G}$ are nonempty sets of geos which we call supergeos. Note that if we restrict all supergeos to have size one, we recover the standard matched pairs design. 

In each supergeo pair we randomly assign one of the supergeos to the treatment group and the other one to the control group. 
Let $A_k \in \{\pm 1\}$ for $k=1,\dots,K$ be $i.i.d.$ random sign variables independent from the rest of the data. Let $A_k=1$ with probability $1/2$ indicating that $G_{k,+}$ is treated and $A_k=-1$ with probability $1/2$ indicating that $G_{k,-}$ is treated. 

For convenience, for a supergeo $G \subset \mathcal{G}$, we define the spend and response of $G$ as $S_{G} := \sum_{g \in G} S_g$ and $R_{G} := \sum_{g \in G} R_g$. Similarly, we define the uninfluenced response of $G$ as $Z_{G} := \sum_{g \in G} Z_g$.

\subsection{Homogeneous $\theta_g$}\label{sec:homo_model}
We start by computing the bias and variance of the empirical estimator, $\hat{\theta}$, under the assumption that all geos respond to treatment in the same way.
\begin{assumption}[Homogeneous $\theta_g$]\label{ass:homogeneous}
For all $g \in \mathcal{G}$, $\theta_g = \theta$.
\end{assumption}

\noindent Note that the response difference between treatment and control of the $k$-th supergeo pair $(G_{k_+}, G_{k_-})$ can be expressed as $A_k (R_{G_{k,+}} - R_{G_{k,-}})$.
Assumptions~\ref{ass:linear} and~\ref{ass:fixed_budget} imply that 
\begin{align*}
     \hat{\theta} & =  \theta +  \frac{1}{B} \sum_{k} A_k (Z_{G_{k,+}} - Z_{G_{k,-}}).
\end{align*}
Since $\{A_k: k=1,\dots,K\}$ are $i.i.d.$~zero-mean random variables that are independent from the rest of the variables, we conclude that $\hat{\theta}$ is an \emph{unbiased} estimator of $\theta$: $\E[\hat{\theta}] = \theta$. The above equation also implies that the variance of $\hat{\theta}$ (note that the randomness is over the $A_k$'s) is:
\begin{align}
    \Var[\hat{\theta}] =\,& \frac{1}{B^2}\sum_{k}(Z_{G_{k,+}} - Z_{G_{k,-}})^2. \label{eq:variance}
\end{align}
The goal of our experimental design procedure is to find supergeo pairs that minimize $\sum_{k}(Z_{G_{k,+}} - Z_{G_{k,-}})^2$. Note that without a constraint on the maximum size of a supergeo pair, the optimal way to minimize the variance is to split all geos into a single supergeo pair $({G_{+}}, {G_{-}})$ with minimum difference $|Z_{G_{+}} - Z_{G_{-}}|$. We will add a size constraint that is important for several reasons that will be discussed in Section~\ref{sec:supergeo_algo}.

\subsection{Heterogeneous $\theta_g$}\label{sec:hetero_model}
In practice, it is quite possible that different geos respond to treatment differently. This leads us to a model in which Assumption~\ref{ass:homogeneous} is abandoned. 
Note that the linearity Assumption~\ref{ass:linear} implies that
\begin{align*}
    \theta = &~ \frac{\sum_g R_g^{(t)} - \sum_g R_g^{(c)}}{\sum_g S_g^{(t)} - \sum_g S_g^{(c)}}
    = \frac{\sum_g \theta_g \cdot (S_g^{(t)} - S_g^{(c)})}{\sum_g (S_g^{(t)} - S_g^{(c)})}.
\end{align*}
Therefore, the ratio of interest, $\theta$, depends on the exact way in which we change the controlled variable $S_g$ under the treatment and control conditions. 
A natural choice is to increase all treatment spends proportionally to the original (pre-experimental) spends and leave the control spends unaffected.
\begin{assumption}[Proportional spend]\label{ass:ideal_model}
Suppose that each geo $g$ has the initial spend $\overline{S}_g$. In the experiment, if $g$ is in the control group, then $S_g^{(c)} = \overline{S}_g$. If $g$ is in the treatment group, then $S_g^{(t)} = \overline{S}_g + \Delta S_g$, where $\Delta S_g = r \cdot \overline{S}_g$. Moreover, the total increase in the spend satisfies the fixed budget assumption: $r\cdot\sum_{g\in\mathcal{T}}\overline{S}_g=B.$
\end{assumption}
\noindent Under this assumption, since we always increase the treatment spend proportionally to $\overline{S}_g$, we have
\begin{equation}\label{eq:theta_hetero_simplify}
\theta = \frac{\sum_g \theta_g \cdot \overline{S}_g}{\sum_g \overline{S}_g}.
\end{equation}
If we conduct the experiment and compute the empirical estimator $\hat{\theta}$ in the same way as before, it is no longer unbiased for $\theta$. Using the linear model assumption and the fixed budget assumption, we have
\begin{align}\label{eq:hetero_error}
    \hat{\theta} 
    = \theta & + \frac{1}{B} \sum_{k} A_k (Z_{G_{k,+}} - Z_{G_{k,-}}) \notag \\
    & + \frac{1}{B} \Big( \sum_{g \in \mathcal{T}} (\theta_g - \theta) \overline{S}_g - \sum_{g \in \mathcal{C}} (\theta_{g} - \theta) \overline{S}_{g} \Big) \notag \\
    & + \frac{1}{B} \sum_{g \in \mathcal{T}} (\theta_g - \theta) \Delta S_g.
\end{align}

\noindent We denote the three error terms in the above equation as $err_1$, $err_2$, and $err_3$. 

$err_1$: The first error term corresponds to the training error which is the same as in the homogeneous model. Since each $A_k$ is uniformly random in $\{-1,1\}$,
\[
\E[err_1] = \E\Big[\sum_{k} A_k (Z_{G_{k,+}} - Z_{G_{k,-}})\Big] = 0.
\]
The magnitude of this error term is small if the (super)geo pairs are well-matched. It can be reduced further by trimming some of the poorly matched pairs. 

$err_2$: The second error term measures the difference between the treatment and control groups in terms of the within-group heterogeneity in $\theta_g$'s. Since each geo is either in the treatment group or the control group with probability $1/2$, we have
\begin{align*}
&~ \E[err_2]
= \E\Big[ \sum_{g \in \mathcal{T}} (\theta_g - \theta) \overline{S}_g - \sum_{g \in \mathcal{C}} (\theta_{g} - \theta) \overline{S}_{g} \Big] \\
= &~ \sum_{g \in \mathcal{G}} \E\Big[\mathbf{1}_{g \in \mathcal{T}} \cdot (\theta_g - \theta) \overline{S}_g - \mathbf{1}_{g \in \mathcal{C}} \cdot (\theta_g - \theta) \overline{S}_g \Big]
= 0.
\end{align*}
The second error term has zero mean for both the supergeo design and the matched pairs design. It remains zero when either all pairs are used or some are removed to improve the matching quality.

$err_3$: The third error term measures the degree of heterogeneity in $\theta_g$'s within the treatment group. The expectation of this term is approximately zero when all geos are included in the experiment. However, this term could induce a bias if some geos are excluded.

By Assumption \ref{ass:fixed_budget}, 
we have
\begin{equation*}
err_3 = \sum_{g \in \mathcal{T}} (\theta_g - \theta) \cdot \Delta S_g = r \cdot \sum_{g \in \mathcal{T}} (\theta_g - \theta) \cdot \overline{S}_g,
\end{equation*}
where $r$ is a random variable that depends on the specific treatment assignment in order to satisfy the fixed budget assumption.

\paragraph{If all geos are included in the experiment.} The key observation is that since $\theta = \sum_g \theta_g \cdot \overline{S}_g / \sum_g \overline{S}_g$,
\begin{equation}\label{eq:hetero}\vspace{-0.2cm}
\sum_g (\theta_g - \theta) \cdot \overline{S}_g = 0.
\end{equation}
Since all geos are included in the experiment, $\sum_{g \in \mathcal{T}} \overline{S}_g \approx 0.5 \cdot \sum_{g \in \mathcal{G}} \overline{S}_g$, so that $r \approx r_0 := 2B/\sum_g \overline{S}_g$. Thus,
\begin{align*}
\E[err_3] = &~ \E\Big[ r \cdot \sum_{g \in \mathcal{T}} (\theta_g - \theta) \cdot \overline{S}_g \Big] \\
\approx &~ r_0 \cdot \E\Big[ \sum_{g \in \mathcal{T}} (\theta_g - \theta) \cdot \overline{S}_g \Big] \\
= &~ r_0 \cdot \sum_g \E[\mathbf{1}_{g \in \mathcal{T}} 
\cdot (\theta_g - \theta) \cdot \overline{S}_g] = 0,
\end{align*}
where the last equality follows from \eqref{eq:hetero} and the fact that each geo has probability $1/2$ to fall into the treatment group.

\paragraph{If some geos are excluded.} In this case, some geos have zero probability to be included in the treatment group and we no longer have $\E[\sum_{g \in \mathcal{T}} (\theta_g - \theta) \cdot \overline{S}_g] = 0$. The bias induced by $err_3$ is especially large when the excluded geos have extra large or extra small values of $\theta_g$ since those differ the most from the weighted average value, $\theta$. 

While trimming poorly matched pairs can reduce the variance of the resulting estimator, it can also induce a substantial bias when $\theta_g$'s of the removed geos differ from the quantity of interest, $\theta$. We provide specific examples illustrating this issue in Section~\ref{sec:eval}. 
\section{Algorithm}\label{sec:supergeo_algo}

Following our theoretical model from Section~\ref{sec:model}, the goal is to find a supergeo design $\{(G_{k,+}, G_{k,-})\}_{k=1}^K$ that minimizes the following loss:
\begin{equation}\label{eq:loss}
\mathrm{loss}(\{(G_{k,+}, G_{k,-})\}_{k=1}^K) := \sum_{k=1}^K (Z_{G_{k,+}} - Z_{G_{k,-}})^2.
\end{equation}
This loss coincides with the variance of the empirical estimator in the homogeneous $\theta_g$ case (Eq.~\eqref{eq:variance}) and corresponds to the variance of term $err_1$ in the heterogeneous $\theta_g$ case. This term is likely to dominate the other two, especially in online advertising applications where ad spend is often at a much smaller scale than revenue.


\paragraph{Matching variables.}
While ideally we would want to match supergeo pairs based on the values of uninfluenced responses, $Z_g$, those are usually unobserved. We follow the existing literature \citep{chen2021trimmed} and approximate the uninfluenced responses by the responses in the pretest phase. Throughout this section we continue using $Z_g$'s to denote the variables we match on, but in practice pretest responses are often used.


\paragraph{Size of supergeo pairs.}
As already discussed, the loss in Eq.~\eqref{eq:loss} is minimized by dividing all geos into two supergeos, i.e.~when $K=1$. However, in practice this is often not desirable for a variety of reasons:
\begin{itemize}
\item {\bf Overfitting.} Since the observed response variables do not have to follow our linear model and since they approximate the true uninfluenced responses which are unobserved,  
using large supergeos could lead to overfitting to the pretest data. In our experiments we observe that while increasing the maximum allowed size of supergeos leads to smaller training losses, the performance on the test set may suffer (see Appendix~\ref{sec:appendix-B-eval} for an example). This issue can be resolved by tuning the maximum allowed size of supergeo pairs---splitting the pretest data into two subsets, one for creating candidate designs and the other for validation---in the spirit of cross-validation used in machine learning. 
\item {\bf Inference.} Using fewer supergeo pairs makes statistical inference harder. For example,
the trimmed match estimator from \citet{chen2022robust} uses $t$-distribution as an approximation to construct the confidence intervals. This approximation requires at least two pairs and the quality of the approximation improves as the number of pairs grows. Similarly, when using inference approaches such as Fisher's exact test \citep{fisher1922interpretation}, we want to allow for many possible treatment assignments as opposed to just two which corresponds to the case of two large supergeos (see Appendix~\ref{sec:inference}).
\item {\bf Robustness.} Increasing the number of supergeo pairs increases the number of possible treatment assignments which makes the experimental design more robust to adversarial considerations (see Appendix~\ref{sec:size} for details). 
\item {\bf Computational efficiency.} Increasing the maximum allowed size of supergeo pairs increases the number of candidate pairs which could be included in the optimal design. Depending on the formulation of the optimization problem, this often leads to the problem being intractable. 
\end{itemize}

\noindent For these reasons, we add a constraint that the size of each supergeo pair is within a range, $[\ell, u]$, i.e.~$\ell \leq |G_{k,+}| + |G_{k,-}| \leq u$ for all $k$, where $\ell$ and $u$ are two tunable parameters. We usually set $\ell = 2$ so that the supergeo design is a natural generalization of the matched pairs design. When reporting the empirical results in Section~\ref{sec:eval}, we limit the size of supergeo pairs to 4 which performs well in our applications.\footnote{We currently do not tune this parameter and set it to a fixed value that provides a good matching quality without leading to particularly large supergeos.}

The first three of the concerns above could also be addressed by adding an explicit ``minimum pairs'' constraint that restricts the number of supergeo pairs to be at least $\kappa$. However, this constraint is ineffective in addressing the issue of computational efficiency. Since the size constraint (combined with the heuristics that we use for the MIP formulation discussed in Appendix~\ref{sec:heuristics}) leads to a design with a sufficient number of supergeo pairs, in our empirical evaluations we set $\kappa = 1$.

\subsection{NP-hardness}
For a supergeo pair $(G_{k,+}, G_{k,-})$, define $G_{k} := G_{k,+} \cup G_{k,-}$. Observe that in order to minimize the loss~\eqref{eq:loss}, $(G_{k,+}, G_{k,-})$ must be an optimal split that minimizes $(Z_{G_{k,+}} - Z_{G_{k,-}})^2$ among all splits of $G_k$:
\[
(G_{k,+}, G_{k,-}) \in {\arg \min}_{\substack{G'_{k,+} \cup G'_{k,-} = G_k \\ G'_{k,+} \cap G'_{k,-} = \emptyset}} (Z_{G'_{k,+}} - Z_{G'_{k,-}})^2.
\]
While there may be multiple optimal splits, $(G_{k,+}, G_{k,-})$, the minimum possible distance between them is uniquely determined given the union $G_k$. Consequently, we can define the loss~\eqref{eq:loss} with respect to $G_k$'s. For any subset $G \in \mathcal{G}$, let
\[
\mathrm{score}(G) := \min_{\substack{G_+ \cup G_- = G \\ G_+ \cap G_- = \emptyset}} (Z_{G_+} - Z_{G_-})^2.
\]
Our goal is to find a partition of $\mathcal{G}$ into disjoint subsets $\{G_k\}_{k=1}^K$ that minimize the loss
\[
\mathrm{loss}(\{G_k\}_{k=1}^K) := \sum_{k=1}^K \mathrm{score}(G_k).
\]



\noindent While the optimal matched pairs design can be found in polynomial time, the proposed generalized matching problem is NP-hard even when the size of $G_{k}$ increases from 2 to 3 (setting $\ell = u = 3$).
\begin{theorem}[NP-hardness]\label{thm:np_hard}
The following problem is NP-hard: Given a set $\mathcal{G}$ of size $|\mathcal{G}|=3m$ and values $Z_g \in \mathbb{Z}^+$ for each $g \in \mathcal{G}$, for any subset $G \subseteq [3m]$ define
\[
\mathrm{score}(G) := \min_{\substack{G_+ \cup G_- = G \\ G_+ \cap G_- = \emptyset}} \Big( \sum_{i \in G_+} Z_{i} - \sum_{j \in G_-} Z_{j} \Big)^2.
\]
Determine whether $\mathcal{G}$ can be partitioned into $m$ disjoint sets $G_1, G_2, \dots, G_m$ such that for all $i \in [m]$, $|G_i|=3$ and $\mathrm{score}(G_i) = 0$ ($[m]$ stands for the set $\{1,2,\dots,m\}$).
\end{theorem}
\noindent Note that this decision problem is easier than the problem we are interested in (loss minimization).
We include a formal proof of the theorem in Appendix~\ref{sec:np_hard}.

\subsection{MIP formulation}\label{sec:mip}
To tackle this NP-hard problem, we propose a covering formulation and solve it using mixed-integer programming (MIP). The idea of using MIP for experimental design is not new \citep[see, for example,][]{zubizarreta2012using,doudchenko2021synthetic,abadie2021synthetic}, but our covering formulation is novel.

Define $\mathcal{F} := \{G \subseteq \mathcal{G} \mid \ell \leq |G| \leq u\}$. Our variable is a boolean vector $x \in \{0,1\}^{|\mathcal{F}|}$ such that component $x_G$ indicates whether a candidate supergeo pair $G$ is included in the supergeo design. 
For each geo $g$, we add a linear constraint that restricts $g$ to be covered by exactly one supergeo pair. To express this constaint, we define a vector $M^{(g)} \in \{0,1\}^{|\mathcal{F}|}$ such that $M^{(g)}_{G} = 1$ if $g \in G$ and $M^{(g)}_{G} = 0$ otherwise. The constraint that geo $g$ is covered by exactly one supergeo pair is equivalent to $(M^{(g)})^{\top} x = 1$. We stack these column-vectors into a matrix $M := [M^{(g_1)} \cdots M^{(g_N)}] \in \{0,1\}^{|\mathcal{F}| \times N}$. 

Let $\mathds{1}_d$ denote the all-one vector of dimension $d$ for any $d$. The mixed-integer program we solve is as follows:
\begin{align*}
    \min_{x}&\;\; \sum_{G \subseteq \mathcal{G}:\ \ell \leq |G| \leq u} \mathrm{score}(G) \cdot x_G \\
    \text{s.t.}&
    ~~~~M^{\top} x = \mathds{1}_N,  \hspace{1.5cm} & \mbox{(Exact cover)}\\
    &~~~~\mathds{1}_{|\mathcal{F}|}^\top \cdot x \geq \; \kappa, &\mbox{(Minimum pairs)} \\
    &~~~~ x \in \{0,1\}^{|\mathcal{F}|}. & \mbox{(Boolean selection)}
\end{align*}
We use SCIP solver of \citet{BestuzhevaEtal2021OO} to solve this MIP which can be done under 1 hour for approximately $N=50$ geos and the maximum size of supergeo pairs equal to 4. To speed up the optimization and make the problem tractable at the practically important scale ($N=210$) we employ a number of heuristics discussed in Appendix~\ref{sec:heuristics}.\footnote{The most straightforward MIP formulation is to use $O(N^2)$ variables such that variable $(i,j)$ indicates whether the $i$-th geo is in the $j$-th supergeo. Our covering formulation has $O(N^u)$ variables. However, it is much easier to solve due to the numerous symmetries between variables intrinsic to the alternative formulation.}

\section{Empirical results}\label{sec:eval}
We evaluate the proposed design in the context of estimating the incremental return on ad spend (iROAS), where for each geographic region the ad spend and the response (e.g.~conversions or sales) are observed. For each geo $g \in \mathcal{G}$, we collect the spend $S_g[1],\dots,S_g[T]$ and responses $R_g[1],\dots,R_g[T]$ over $T$ time periods. We divide $T$ into a pretest phase $\{1,\dots,T_0\}$ and a test phase $\{T_0+1,\dots,T\}$, and refer to the aggregates $R_g^{\mathrm{pre}} := \sum_{t=1}^{T_0} R_g[t]$ and $R_g^{\mathrm{test}} := \sum_{t=T_0 + 1}^{T} R_g[t]$ as the (total) pretest and test responses respectively. We define the pretest spend $S_g^{\mathrm{pre}}$ and the test spend $S_g^{\mathrm{test}}$ similarly. We use the pretest response $R_g^{\mathrm{pre}}$ as the approximation of $Z_g$ which is unobserved and we use the algorithms described in Section~\ref{sec:supergeo_algo} to generate experimental designs. The tuning parameters $\ell$ (minimum size of a supergeo pair), $u$ (maximum size of a supergeo pair), and $\kappa$ (minimum number of supergeo pairs) are fixed at 2, 4, and 1 respectively.

We iterate over different treatment assignments performing what we call the ``half-synthetic'' evaluations---the underlying data comes from real online advertising settings with artificial treatment ``injected'' in the data. Specifically, at each iteration we increase (\emph{``heavy-up''}) the spend in each treated geo in the test period to $(1 + r) \cdot S_g^{\mathrm{test}}$ and leave the spend in the control geos unaffected.\footnote{The methods of this paper can also be applied to another commonly used type of experiment in which the spend in treated geos is set to zero (\emph{``go dark''}).} The response variable in a treated geo $g$ is increased to $R_g^{\mathrm{test}} + \theta_g \cdot r S_g^{\mathrm{test}}$, where $\theta_g$ is the iROAS for that geo that is chosen by us and depends on the particular evaluation as discussed further. The ratio $r$ is set in such a way that the total increase in the spend is equal to the fixed budget, $B$.

The underlying data comes from one of the two real datasets that we call ``A'' and ``B.'' In both datasets $\mathcal{G}$ corresponds to the set of 210 DMAs in the United States, and the unit of time is one day. We set both the pretest and test phases to be four weeks long. The budget, $B$, is chosen to roughly match the magnitude of the total spend in the absence of treatment effects. 
The true response and spend values are scaled to be between $[0,1]$. For all experiments reported in this paper, the evaluation procedures are iterated $M=10000$ times across different treatment assignments. 
We run the heuristics from Appendix~\ref{sec:heuristics} with different parameters in parallel with a time limit of 3 hours, and then pick the design that leads to the smallest training loss. 

We compare the supergeo design to matched pairs design using the empirical estimator~(Eq.~\eqref{eq:theta_hat}) as well as the trimmed match estimator~(Eq.~\eqref{eq:theta_hat_trim}). The poorly matched pairs can be trimmed during either the design stage or the analysis stage (or both). This has the potential to greatly reduce the variance of the resulting estimator, but may introduce substantial bias if the trimmed geos differ from the remaining ones in terms of their response to treatment. In the evaluations we have conducted, the two approaches---trimming at the design stage or at the analysis stage---led to similar results with the latter being generally more effective since it can take into account not only the similarity of geos in the pretest period, but also their potential divergence in the test period due to noise. Consequently, we present the comparisons between the supergeo design which utilizes a simple empirical estimator and the matched pairs design which utilizes the trimmed match estimator from \citet{chen2022robust}. 
However, the takeaways from the analysis would remain similar if the trimming was done at the design phase. We also emphasize that while the main advantage of the supergeo design is the ability to improve the matching quality without sacrificing any data, this design can be combined with trimming. This leads to a performance that is similar or better than that of the trimmed match estimator applied to matched pairs design, but may still lead to a bias when the effects are heterogeneous. 

\subsection{Homogeneous $\theta_g$}\label{sec:evaluation_homogeneous}
First, we present the comparisons between the different methodologies for the case of homogeneous $\theta_g=1$ for all $g\in\mathcal{G}$. In Table~\ref{tab:homogeneous} we report the \emph{root-mean-square error} (RMSE), $\sqrt{M^{-1}\sum_{m=1}^M (\hat{\theta}^{(m)} - \theta)^2}$, and the absolute bias, $|M^{-1}\sum_{m=1}^M \hat{\theta}^{(m)} - \theta|$, where $\hat{\theta}^{(m)}$ corresponds to the estimate computed in iteration $m$. 
See Figure~\ref{fig:homogeneous} in Appendix~\ref{sec:other_figs} for the histograms of the estimates.

\begin{table}[!h]
\centering
\begin{tabular}{l|l|cc|cc}
    \toprule
    \multirow{2}{*}{Est.} & \multirow{2}{*}{Design} & \multicolumn{2}{c|}{Dataset A} & \multicolumn{2}{c}{Dataset B} \\
    \cmidrule(r){3-6}
    & & RMSE & Bias & RMSE & Bias \\
    \midrule
    \multirow{2}{*}{$\hat{\theta}$} & Pairs & $0.97$ & $0.022$ & $2.03$ & $0.107$ \\
     & Supergeo & $0.42$ & $0.005$ & $0.30$ & $0.006$ \\
    \midrule
    \multirow{2}{*}{$\hat{\theta}^{\mathrm{trim}}$} & Pairs & $0.33$ & $0.006$ & $0.42$ & $0.003$ \\
     & Supergeo & $0.34$ & $0.006$ & $0.37$ & $0.016$ \\
    \bottomrule
  \end{tabular}
  
\medskip
\raggedright
{\small\textit{Note}: The estimator (abbreviated as Est.) is either the empirical estimator $\hat{\theta}$ or the trimmed match estimator $\hat{\theta}^{\mathrm{trim}}$. The experimental design is either the matched pairs design (abbreviated as Pairs) 
or the supergeo design (abbreviated as Supergeo).}
\caption{Empirical results under homogeneous iROAS: the RMSE and the bias of the estimates.}
\label{tab:homogeneous}
\end{table}

As evident from the table, supergeo design almost uniformly outperforms the matched pairs design in terms of RMSE when both designs use the same estimator.\footnote{The only exception being the trimmed match estimator on dataset A, where the supergeo design is very slightly worse than the matched pairs design.} Moreover, the performance of the empirical estimator applied to supergeo design is comparable to that of the trimmed match estimator applied to the matched pairs design (supergeo design is better on dataset B, and worse on dataset A)---supergeo design is able to achieve similar performance without throwing out any data.

In Figure~\ref{fig:design_small} we report the matching quality achieved by the two designs in the pretest period. 
We can see that supergeo design achieves a much better matching quality than the matched pairs design on the pretest data. Figure~\ref{fig:design} in Appendix~\ref{sec:other_figs} shows that on the test data, the matching quality of the supergeo design declines relative to the pretest period due to temporal noise. We also observe that this noise is more substantial in dataset A compared to dataset B, and we conjecture that this is the reason behind supergeo design performing worse on dataset A.

\begin{figure}[!tb]
\centering
\subfigure[Matched pairs design]{\label{fig:A_baseline_design_small}{\includegraphics[trim={0 18cm 2.5cm 3.9cm},clip,width=0.48\textwidth]{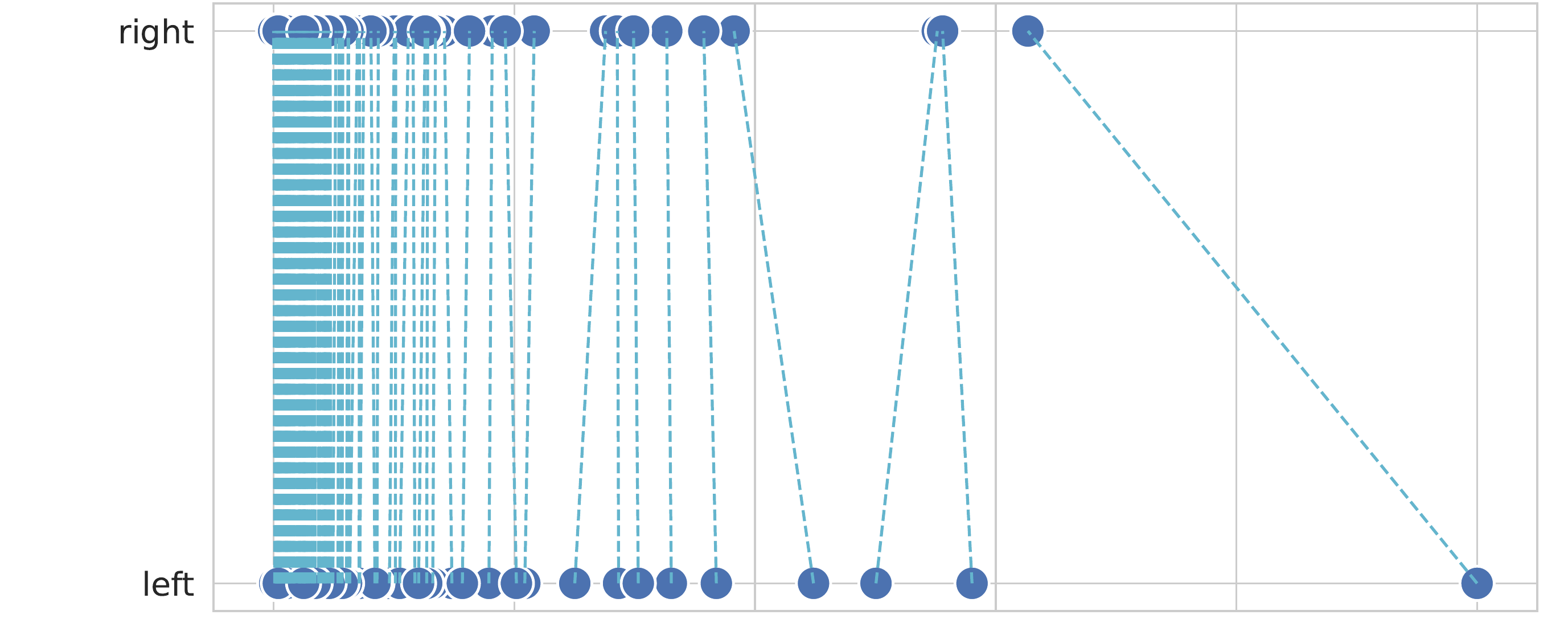}}}
\subfigure[Supergeo design]{\label{fig:A_supergeo_design_small}{\includegraphics[trim={0 18cm 2.5cm 3.9cm},clip,width=0.48\textwidth]{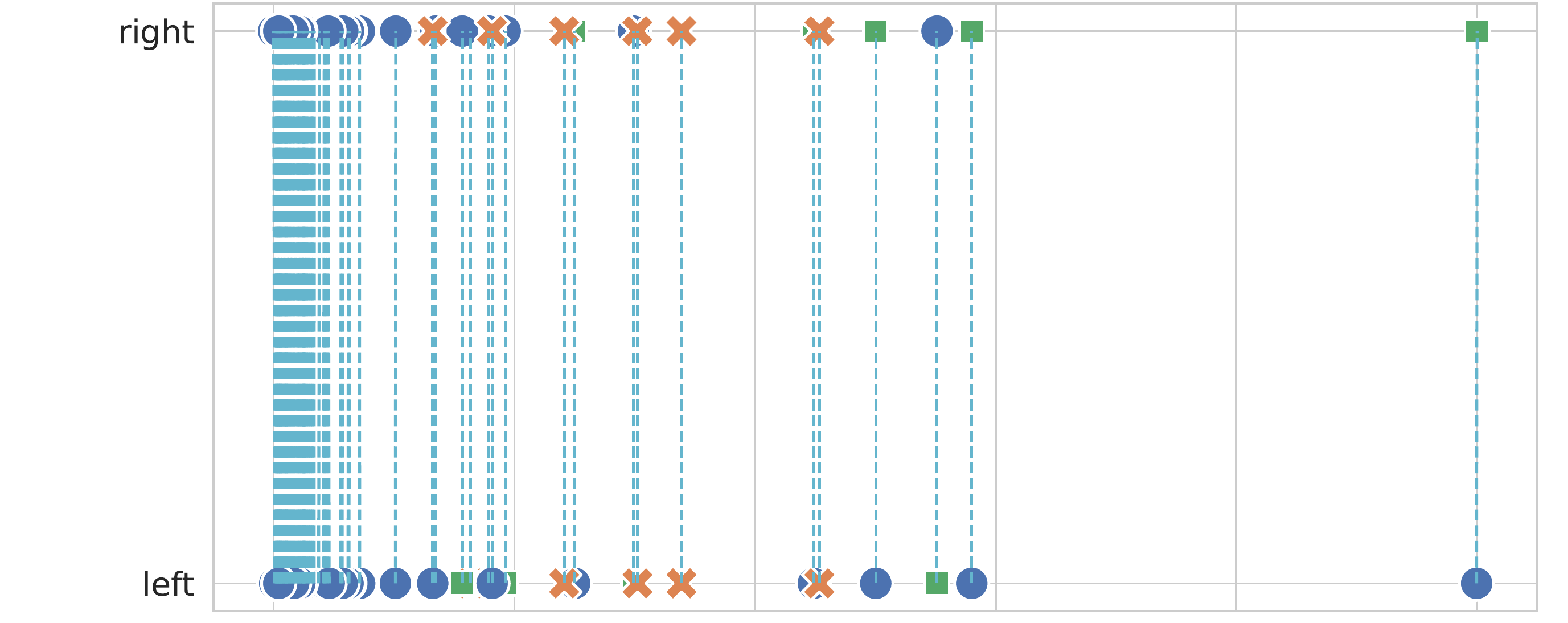}}}

\medskip
\raggedright
{\small\textit{Note}: Each dot represents one (super)geo and the dashed line between two dots implies that these two (super)geos are matched with ``left'' and ``right'' representing the two ``sides'' of each (super)geo pair. The $x$-axis is the scaled value of total pretest response. The blue circle represents a single geo, the orange cross represents a supergeo consisting of 2 geos, the green square represents a supergeo consisting of 3 geos.}
\caption{Matching quality achieved by different designs on dataset A in the pretest phase.
}
\label{fig:design_small}
\end{figure}

\subsection{Heterogeneous $\theta_g$}
Second, we consider the case when $\theta_g$'s can vary across geos. We modify $\theta_g$ by adding a compotent that is proportional to the uninfluenced response of a geo: $\theta_g = 1 + c \cdot Z_g/(N^{-1}\sum_{g} Z_g)$, where $c=0.2$ when applied to dataset A and $0.07$ when applied to dataset B. Following Eq.~\eqref{eq:theta_hetero_simplify}, the quantity of interest is computed as $\theta = \sum_g\theta_g \cdot S_g^{\mathrm{test}}/\sum_g S_g^{\mathrm{test}}$. 

This model captures a scenario in which advertising is more effective in densely populated areas that often generate larger revenues. This assumption may or may not be adequate in a particular setting. What is important, is the fact that trimming of the poorly matched pairs will likely be non-random with respect to the distribution of $\theta_g$'s across geos. In Appendix~\ref{sec:other_hetero} we report the results from another simulation study with heterogeneous $\theta_g$'s---each $\theta_g$ equals $1$ plus a uniformly random noise. Under this model, the results are similar to those from the homogeneous model---while the sample size is reduced, the geos are trimmed in a way that is random with respect to the distribution of $\theta_g$'s, resulting in no substantial bias. This confirms our intuition from Section~\ref{sec:hetero_model}. 

Table~\ref{tab:heterogeneous} reports the RMSE and the absolute bias while Figure~\ref{fig:heterogeneous} in Appendix~\ref{sec:other_figs} shows the histograms of the estimates. Supergeo design---when combined with the empirical estimator---reduces the variance without introducing a bias and leads to the most precise estimates across the board. The trimmed match estimator is effective at reducing the RMSE of the matched pairs design by reducing the variance, but that is achieved at the cost of introducing a bias which may be substantial. Moreover, it does not improve the results over the empirical estimator when the supergeo design is used---the potential reduction in variance is not worth the price paid in the increased bias.


\begin{table}[!h]
\centering
\begin{tabular}{l|l|cc|cc}
    \toprule
    \multirow{2}{*}{Est.} & \multirow{2}{*}{Design} & \multicolumn{2}{c|}{Dataset A} & \multicolumn{2}{c}{Dataset B} \\
    \cmidrule(r){3-6}
    & & RMSE & Bias & RMSE & Bias \\
    \midrule
    \multirow{2}{*}{$\hat{\theta}$} & Pairs & $1.01$ & $0.040$ & $2.17$ & $0.140$ \\
     & Supergeo & $0.38$ & $0.003$ & $0.47$ & $0.004$ \\
    \midrule
    \multirow{2}{*}{$\hat{\theta}^{\mathrm{trim}}$} & Pairs & $0.58$ & $0.429$ & $0.65$ & $0.507$ \\
     & Supergeo & $0.54$ & $0.442$ & $0.52$ & $0.420$ \\
    \bottomrule
  \end{tabular}
  
\medskip
\centering
{\small\textit{Note}: The table is shown using the same format as is used in Table~\ref{tab:homogeneous}.}
\caption{Empirical results under heterogeneous iROAS.}
\label{tab:heterogeneous}
\end{table}

\section{Future work}
There are a number of potential future extensions:
\begin{itemize}
  \item Running experiments on half of the entire population is expensive. If the treatment fraction could be reduced, the cost savings would often surpass the potential loss in statistical power. This can be achieved, for example, by constructing supergeo triplets, quadruplets, etc.~and assigning only one supergeo in each group to treatment. 
  \item The current MIP formulation can be extended to include additional matching variables such as: socio-demographic characteristics, distance between geos, etc. Geo aggregation will need to be adjusted accordingly.
  \item The current design does not explicitly account for proximity of the regions, or travel patterns introducing additional interference/contamination. These considerations can be taken into account when constructing supergeos.
\end{itemize}



\clearpage
\newpage
\bibliography{ref}
\bibliographystyle{icml2023}

\newpage
\appendix
\onecolumn
\begin{figure*}[!ht]
\centering
\subfigure[Dataset A: Matched pairs design]{\label{fig:A_baseline_design}{\includegraphics[trim={0 2cm 2.5cm 2cm},clip,width=0.44\textwidth]{figs/old_pilot_baseline/design_plot.pdf}}}
\subfigure[Dataset A: Supergeo design]{\label{fig:A_supergeo_design}{\includegraphics[trim={0 2cm 2.5cm 2cm},clip,width=0.44\textwidth]{figs/old_pilot/design_plot.pdf}}}
\subfigure[Dataset B: Matched pairs design]{\label{fig:B_baseline_design}{\includegraphics[trim={0 2cm 2.5cm 2cm},clip,width=0.44\textwidth]{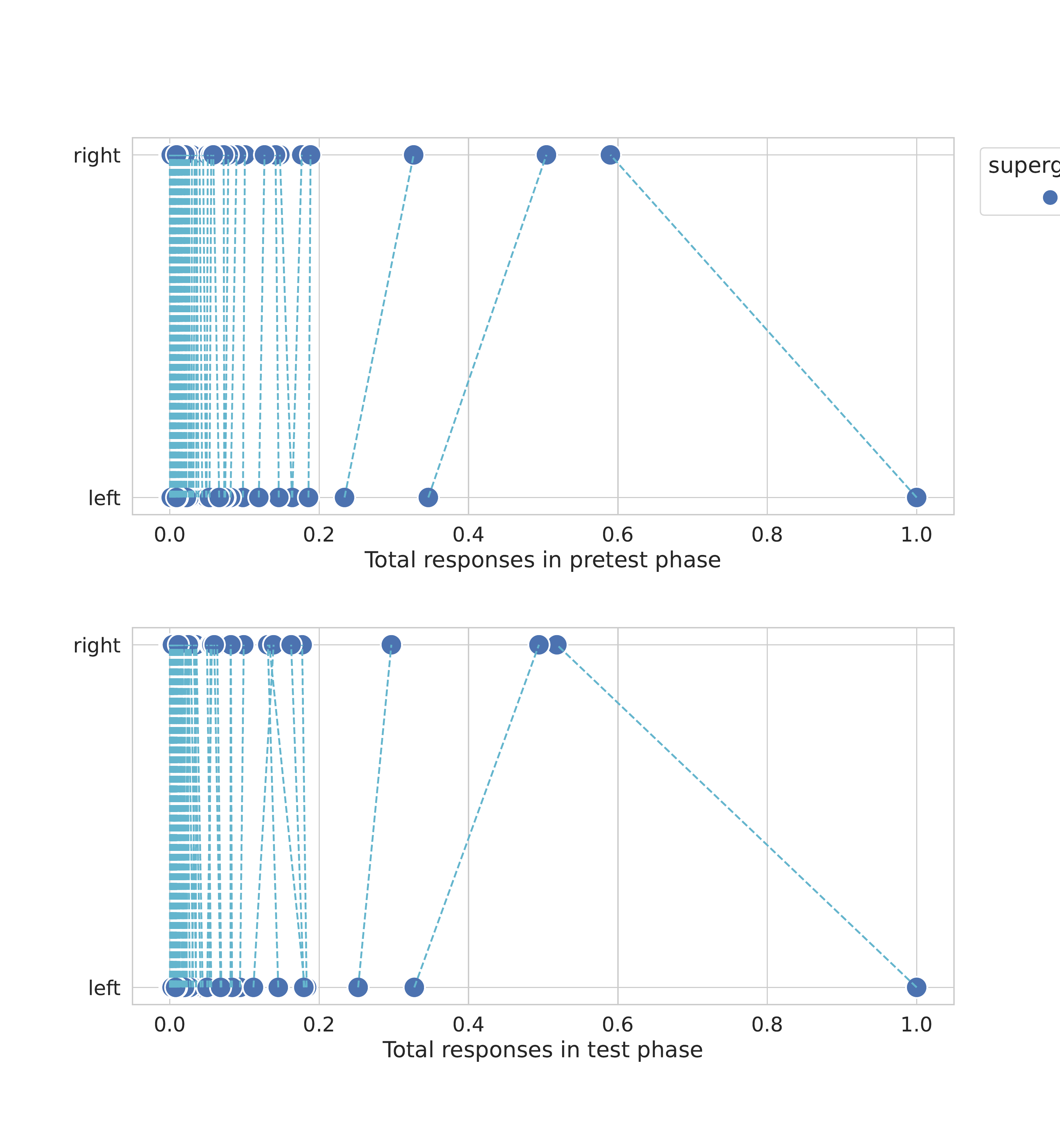}}}
\subfigure[Dataset B: Supergeo design]{\label{fig:B_supergeo_design}{\includegraphics[trim={0 2cm 2.5cm 2cm},clip,width=0.44\textwidth]{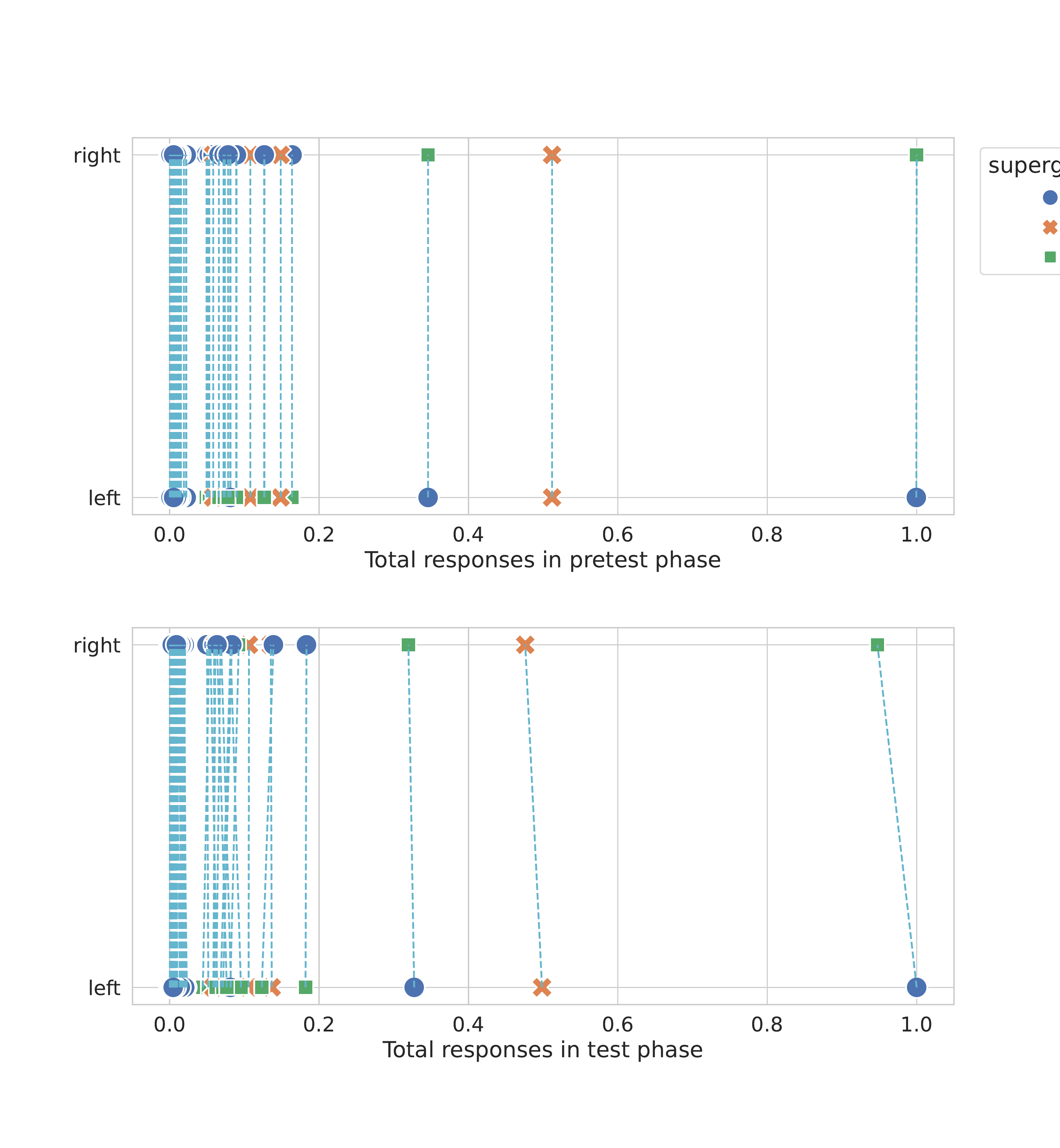}}}

\medskip
\raggedright
{\small\textit{Note}: Each dot represents one (super)geo and the dashed line between two dots implies that these two (super)geos are matched with ``left'' and ``right'' representing the two ``sides'' of each (super)geo pair. The $x$-axis is the scaled value of total pretest or test response. The blue circle represents a single geo, the orange cross represents a supergeo consisting of 2 geos, the green square represents a supergeo consisting of 3 geos.}
\caption{Matching quality achieved by different designs on datasets A and B.}
\label{fig:design}
\end{figure*}


\begin{figure*}[!ht]
\centering
\begin{tabular}{@{} c|c @{}}
\subfigure[A: $\hat{\theta}$ of Pairs ]{\label{fig:A_no_trim}{\includegraphics[width=0.24\textwidth]{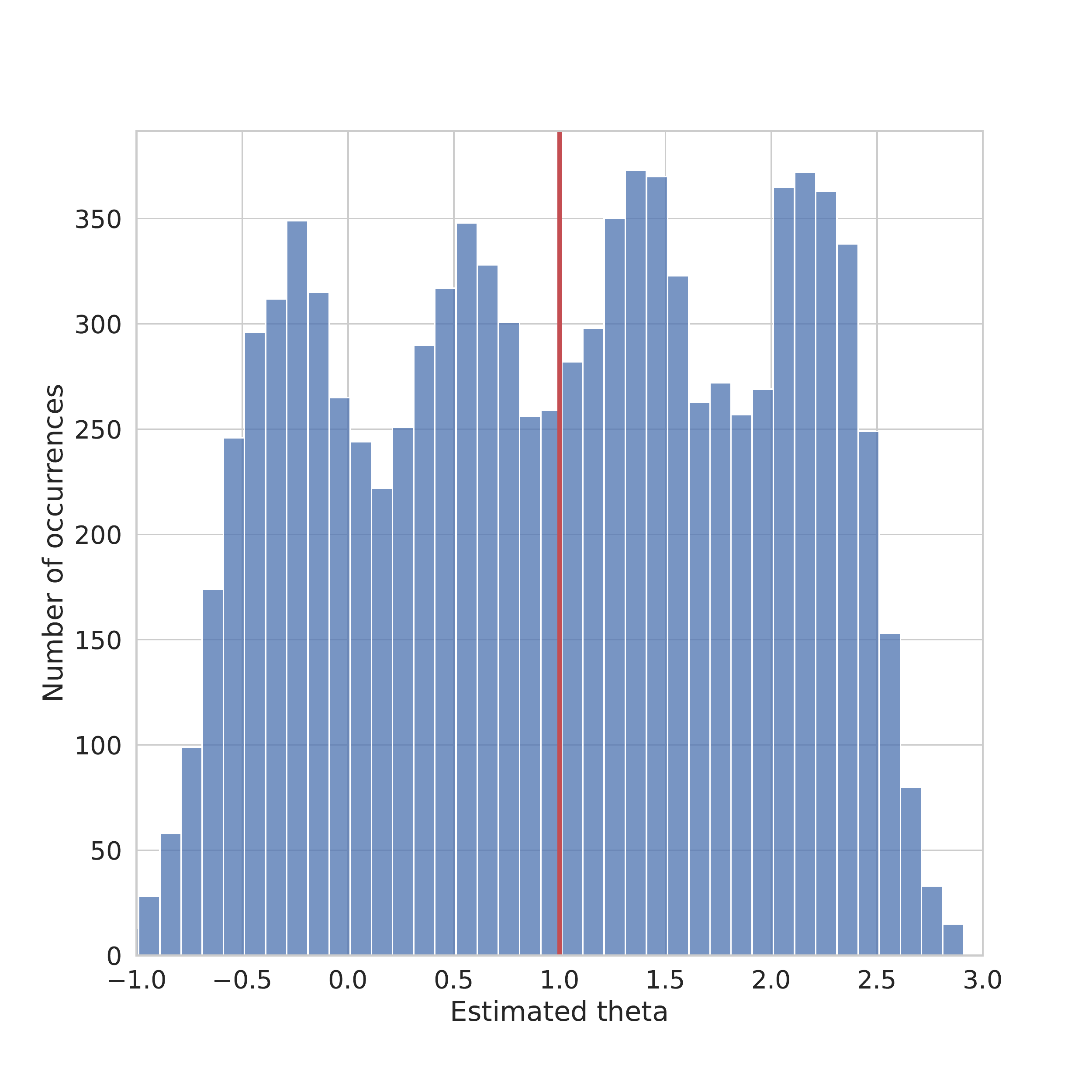}}}
\subfigure[A: $\hat{\theta}$ of Supergeo]{\label{fig:A_supergeo}{\includegraphics[width=0.24\textwidth]{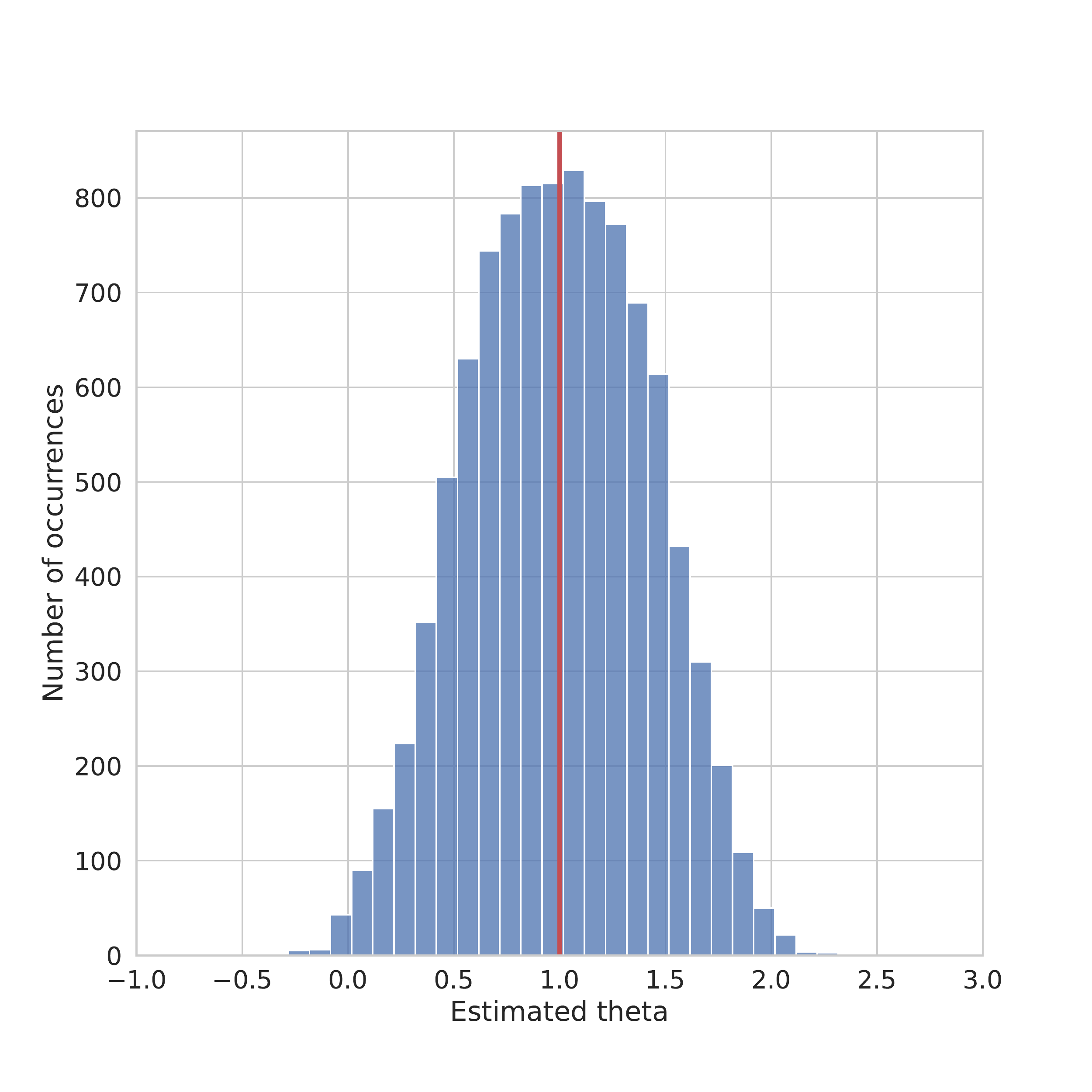}}}
&
\subfigure[A: $\hat{\theta}^{\mathrm{trim}}$ of Pairs ]{\label{fig:A_with_trim}{\includegraphics[width=0.24\textwidth]{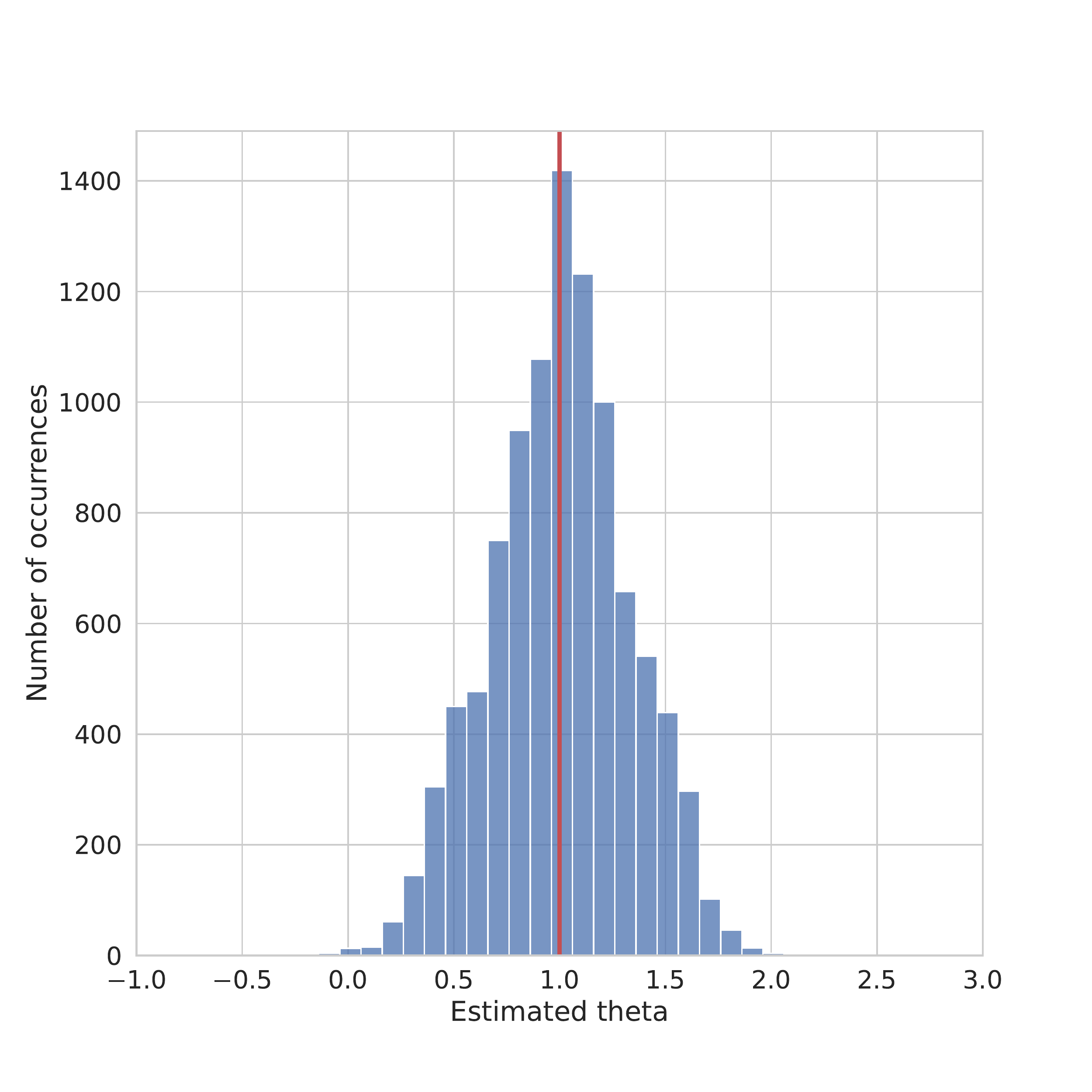}}}
\subfigure[A: $\hat{\theta}^{\mathrm{trim}}$ of Supergeo]{\label{fig:A_supergeo_trim}{\includegraphics[width=0.24\textwidth]{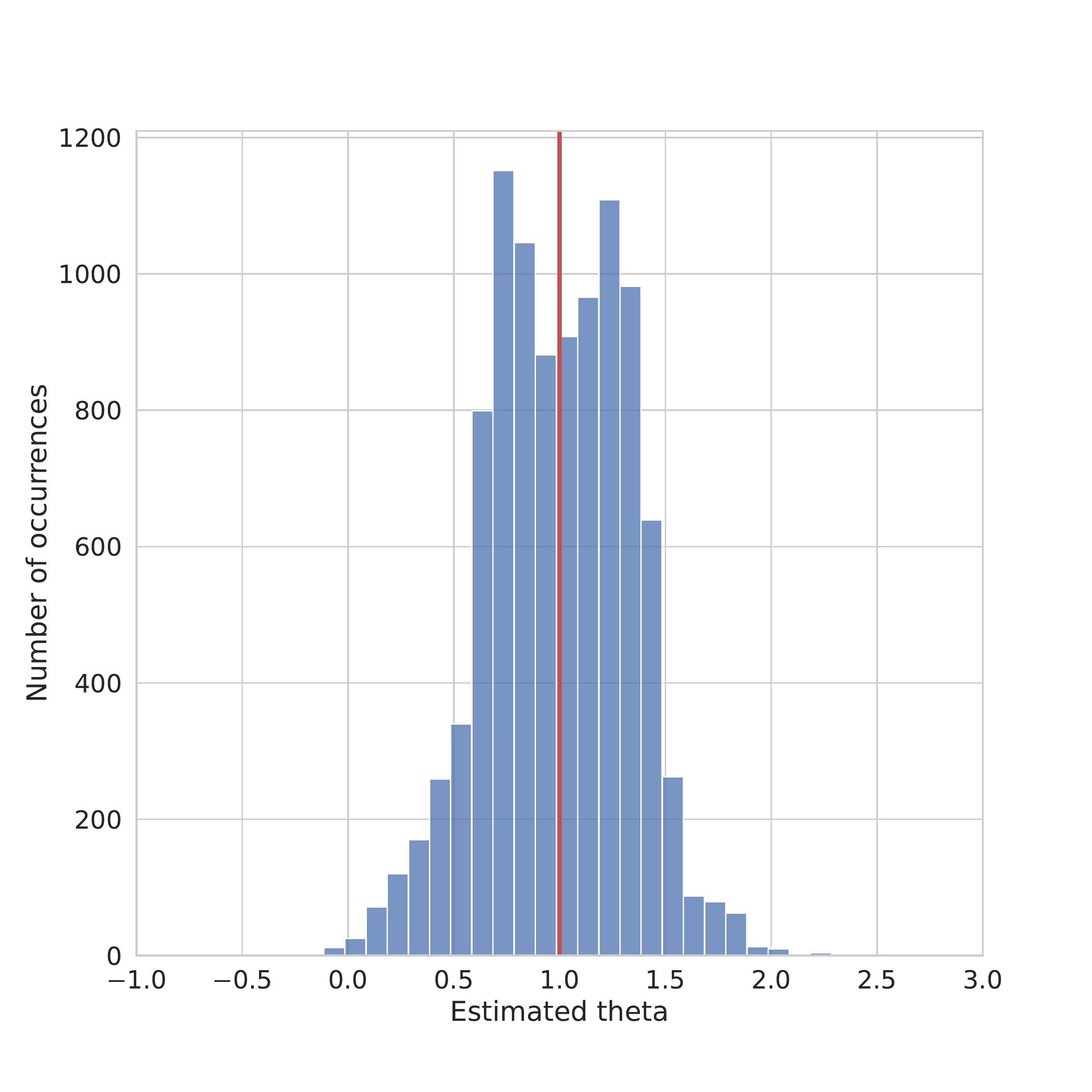}}}
\\ \hline
\subfigure[B: $\hat{\theta}$ of Pairs]{\label{fig:B_no_trim}{\includegraphics[width=0.24\textwidth]{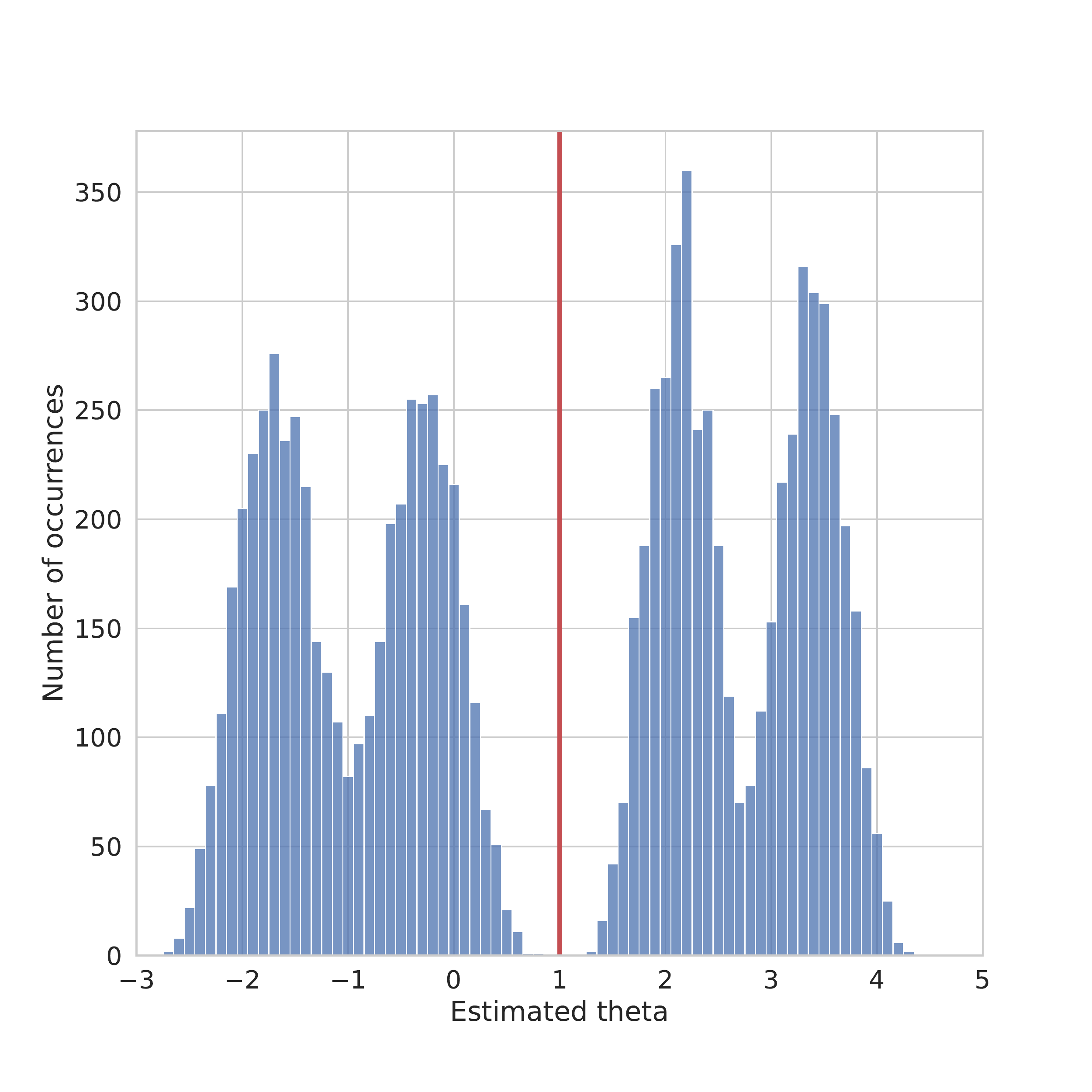}}}
\subfigure[B: $\hat{\theta}$ of Supergeo]{\label{fig:B_supergeo}{\includegraphics[width=0.24\textwidth]{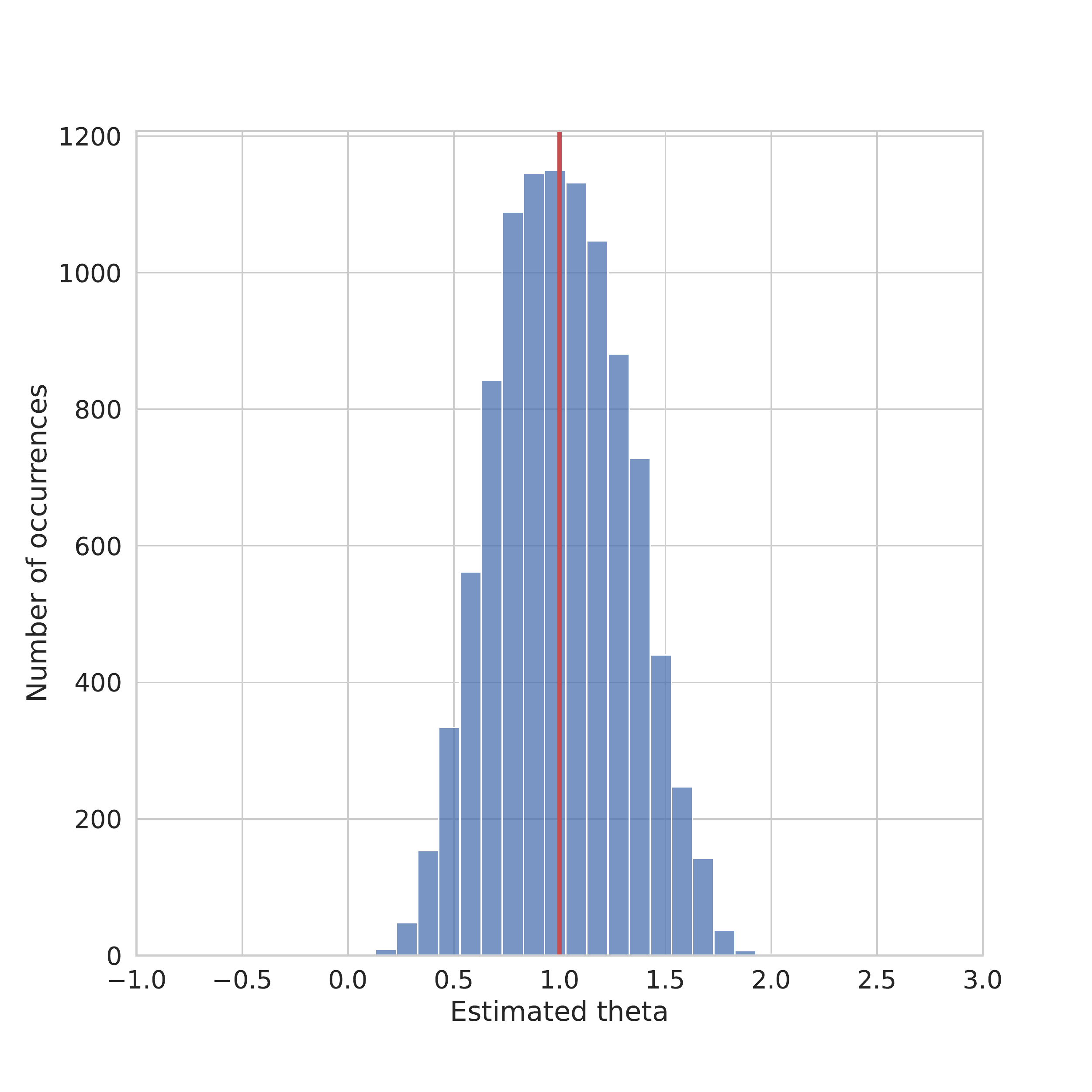}}}
&
\subfigure[B: $\hat{\theta}^{\mathrm{trim}}$ of Pairs ]{\label{fig:B_with_trim}{\includegraphics[width=0.24\textwidth]{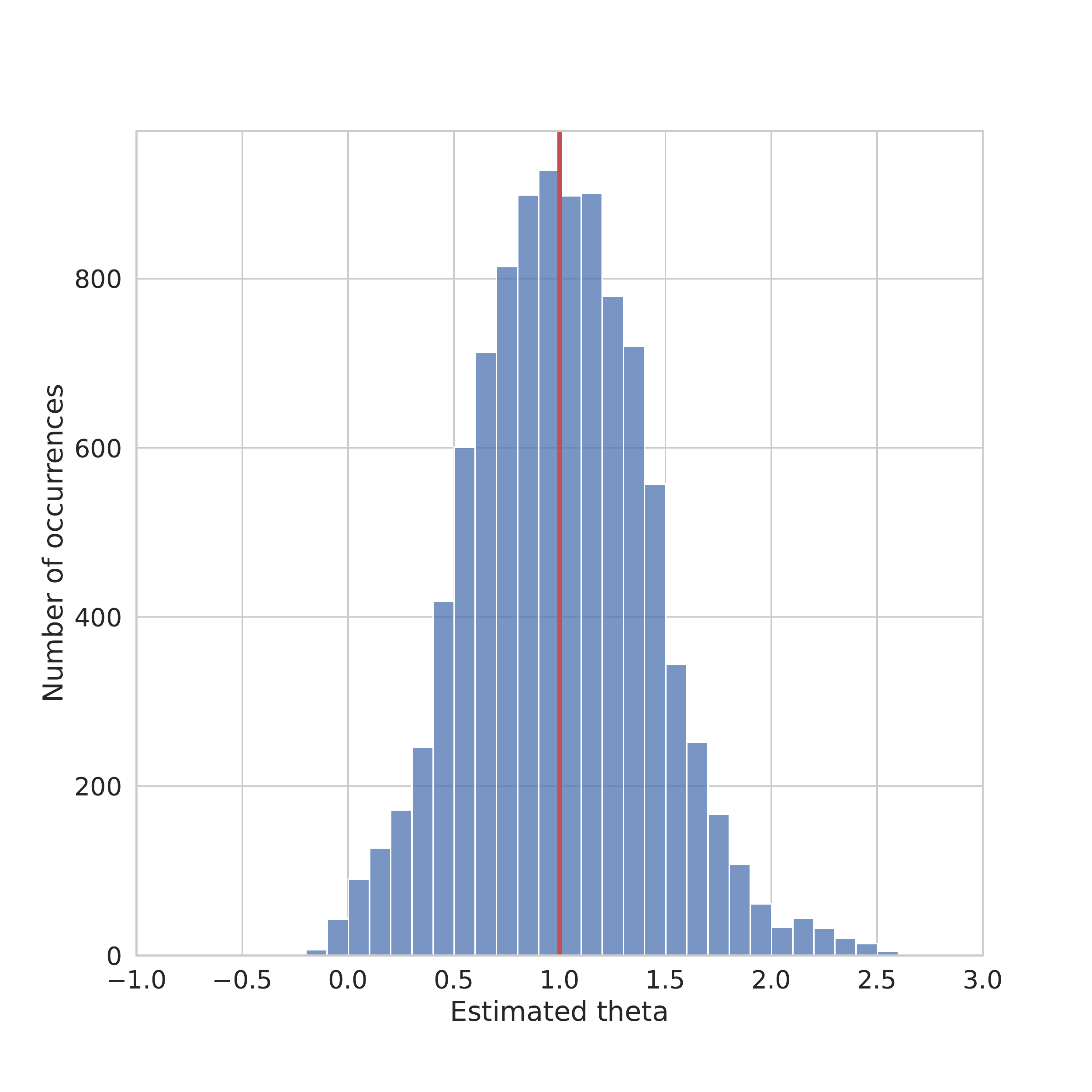}}}
\subfigure[B: $\hat{\theta}^{\mathrm{trim}}$ of Supergeo]{\label{fig:B_supergeo_trim}{\includegraphics[width=0.24\textwidth]{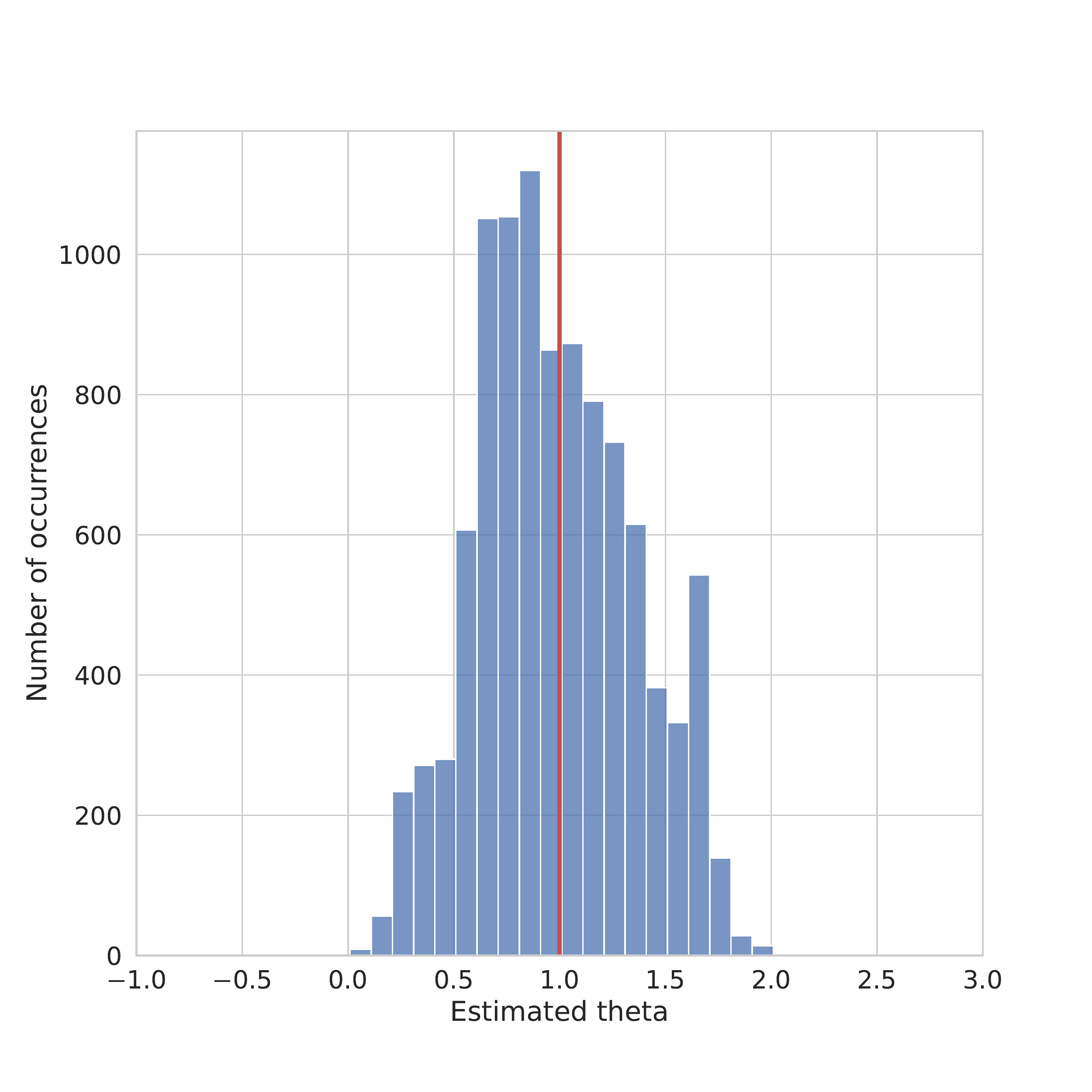}}}
\end{tabular}

\medskip
\raggedright
{\small\textit{Note}: The dataset is either A or B, the estimator is either the empirical estimator $\hat{\theta}$ or the trimmed match estimator $\hat{\theta}^{\mathrm{trim}}$, and the experimental design is either the matched pairs design (abbreviated as Pairs) or the supergeo design (abbreviated as Supergeo). The red vertical line corresponds to the true value of $\theta$.}
\caption{The histograms of the estimates under homogeneous iROAS.}
\label{fig:homogeneous}
\end{figure*}

\begin{figure*}[!ht]
\centering
\begin{tabular}{@{} c|c @{}}
\subfigure[A: $\hat{\theta}$ of Pairs]{\label{fig:A_hetero_no_trim}{\includegraphics[width=0.24\textwidth]{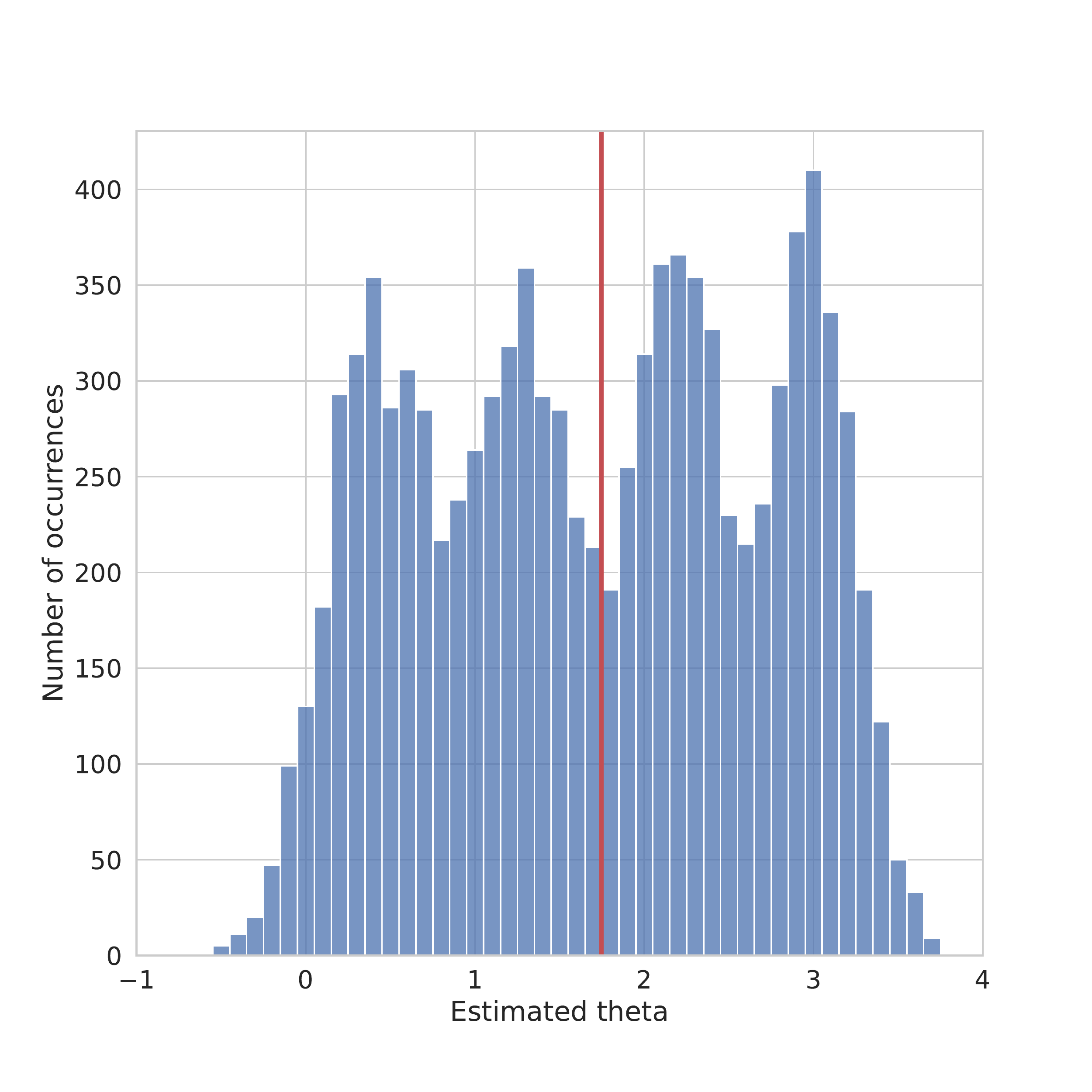}}}
\subfigure[A: $\hat{\theta}$ of Supergeo]{\label{fig:A_hetero_supergeo}{\includegraphics[width=0.24\textwidth]{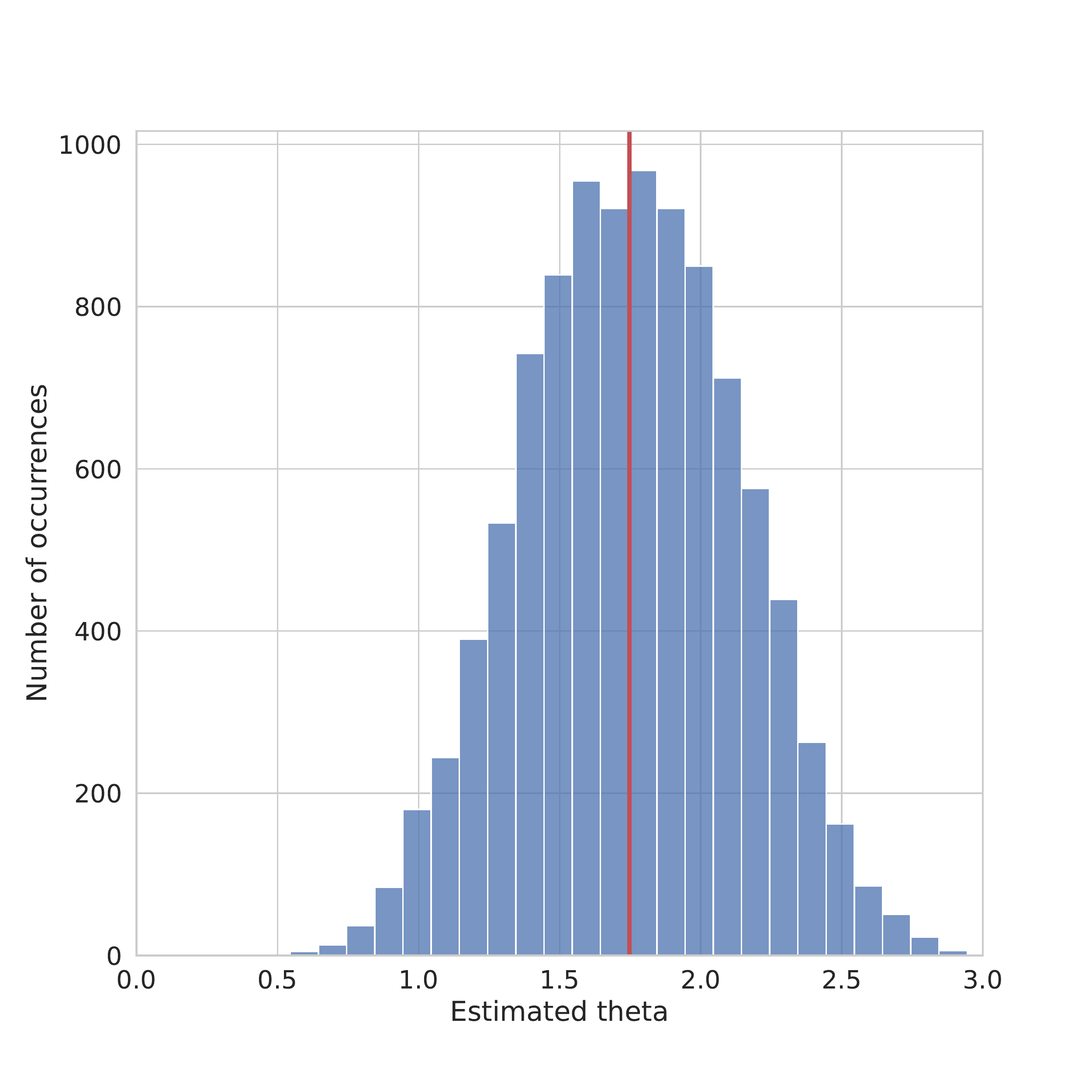}}}
&
\subfigure[A: $\hat{\theta}^{\mathrm{trim}}$ of Pairs ]{\label{fig:A_hetero_with_trim}{\includegraphics[width=0.24\textwidth]{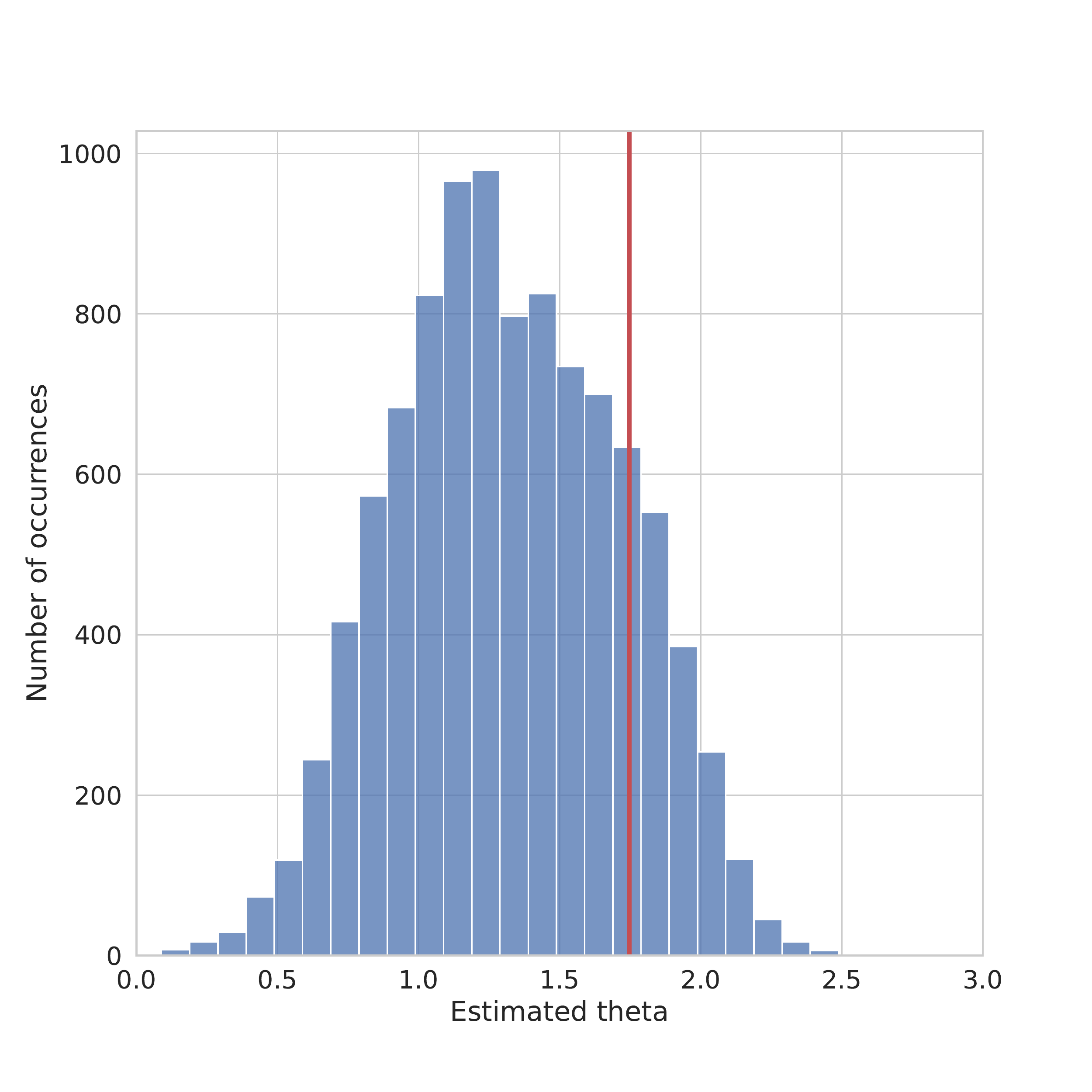}}}
\subfigure[A: $\hat{\theta}^{\mathrm{trim}}$ of Supergeo]{\label{fig:A_hetero_supergeo_trim}{\includegraphics[width=0.24\textwidth]{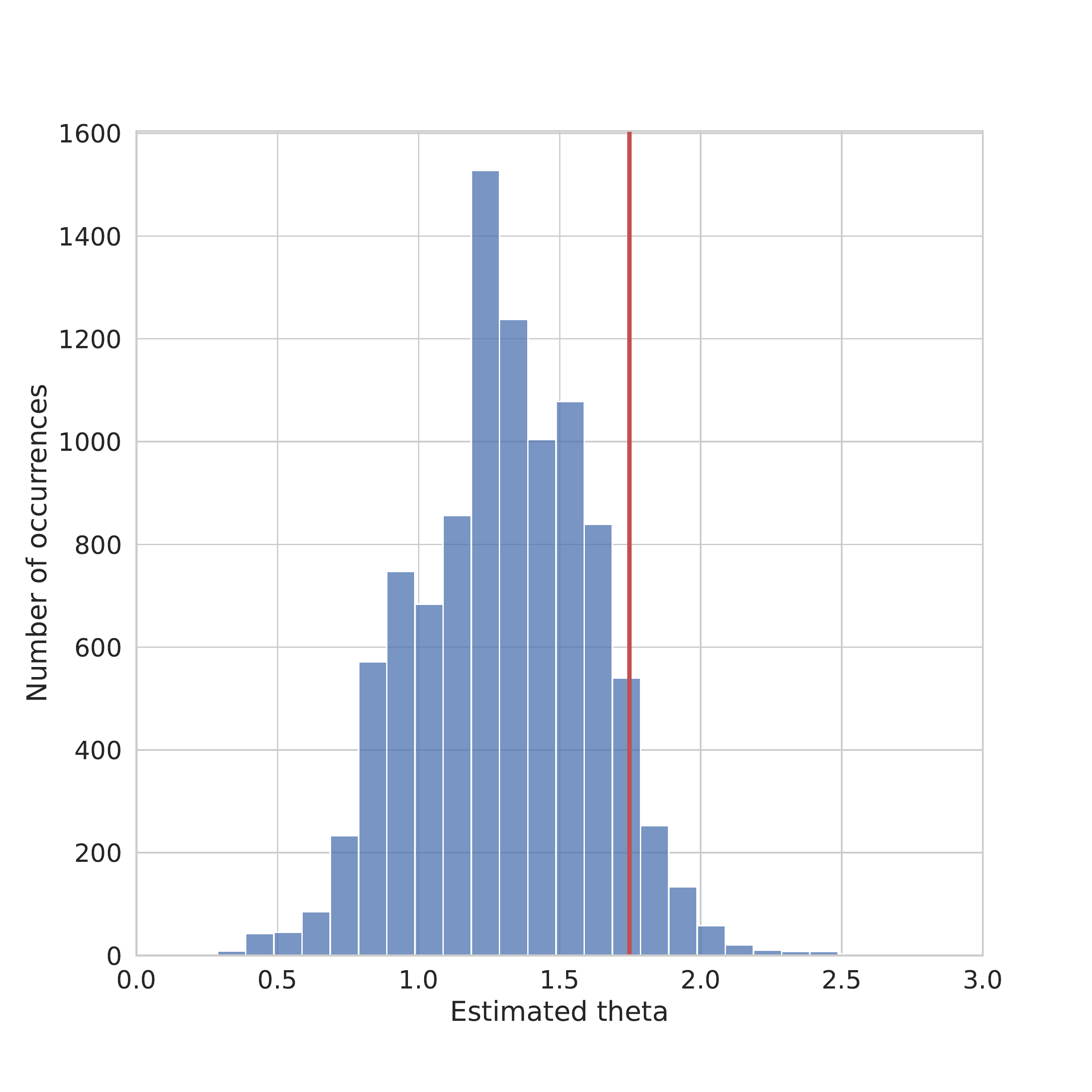}}}
\\ \hline
\subfigure[B: $\hat{\theta}$ of Pairs]{\label{fig:B_hetero_no_trim}{\includegraphics[width=0.24\textwidth]{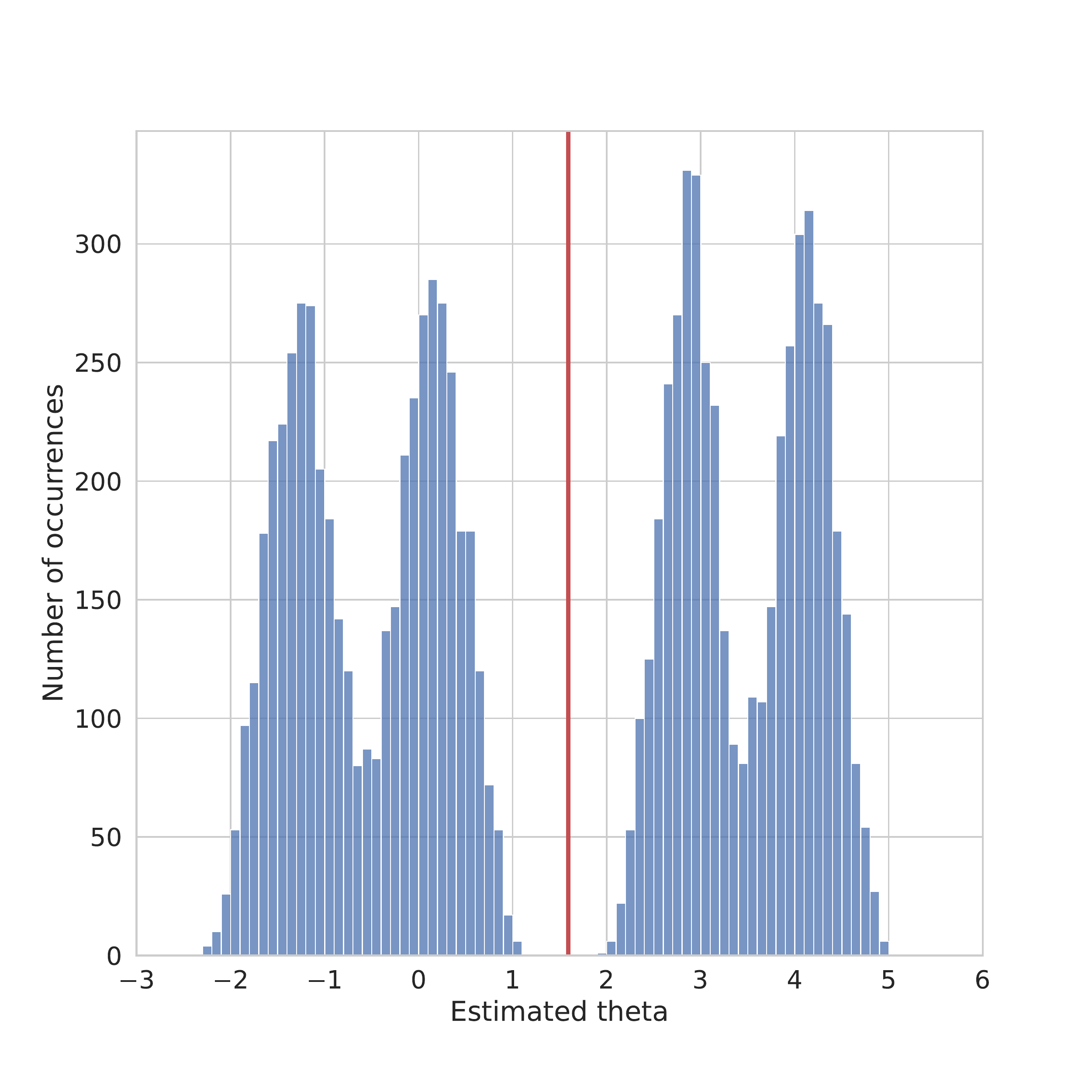}}}
\subfigure[B: $\hat{\theta}$ of Supergeo]{\label{fig:B_hetero_supergeo}{\includegraphics[width=0.24\textwidth]{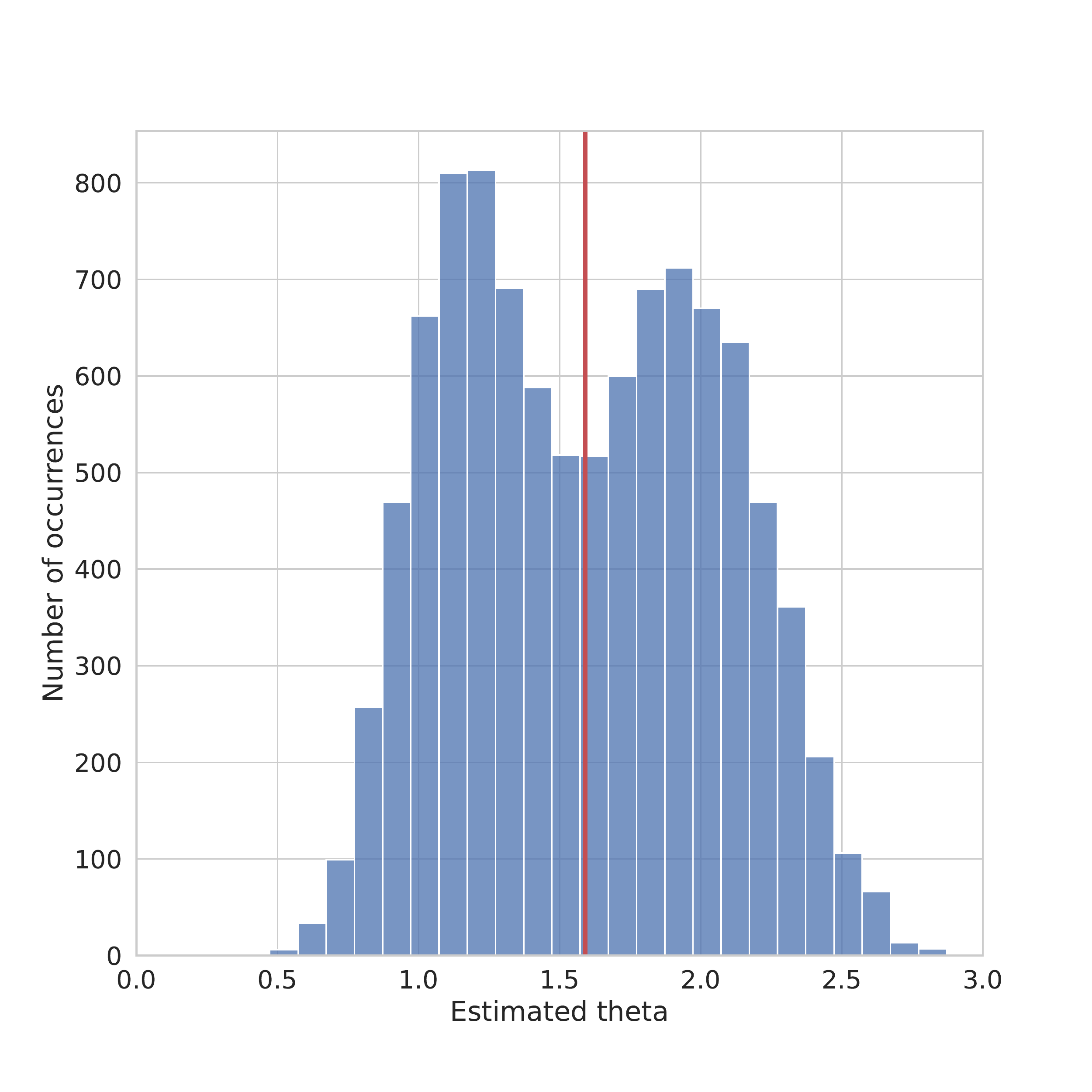}}}
&
\subfigure[B: $\hat{\theta}^{\mathrm{trim}}$ of Pairs ]{\label{fig:B_hetero_with_trim}{\includegraphics[width=0.24\textwidth]{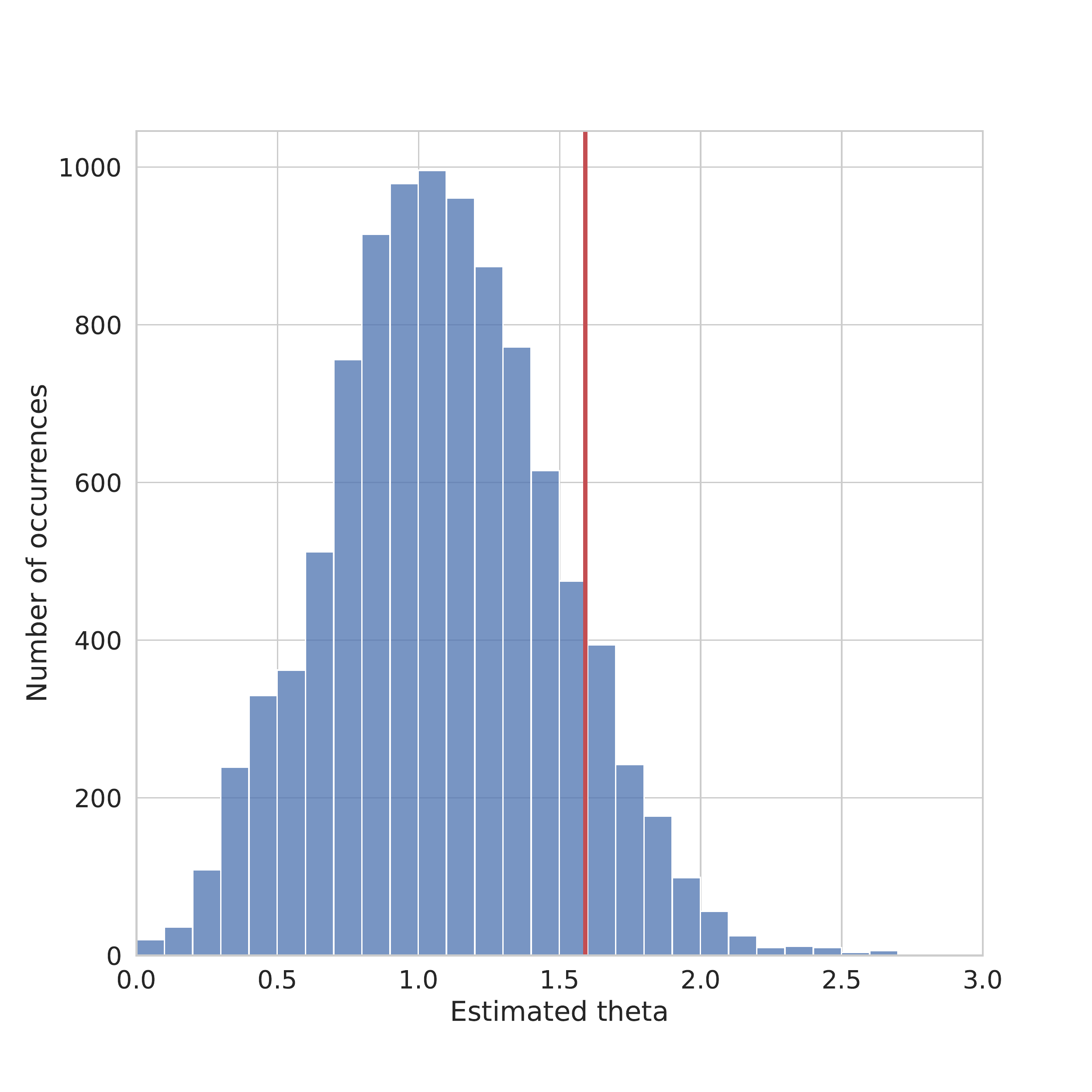}}}
\subfigure[B: $\hat{\theta}^{\mathrm{trim}}$ of Supergeo]{\label{fig:B_hetero_supergeo_trim}{\includegraphics[width=0.24\textwidth]{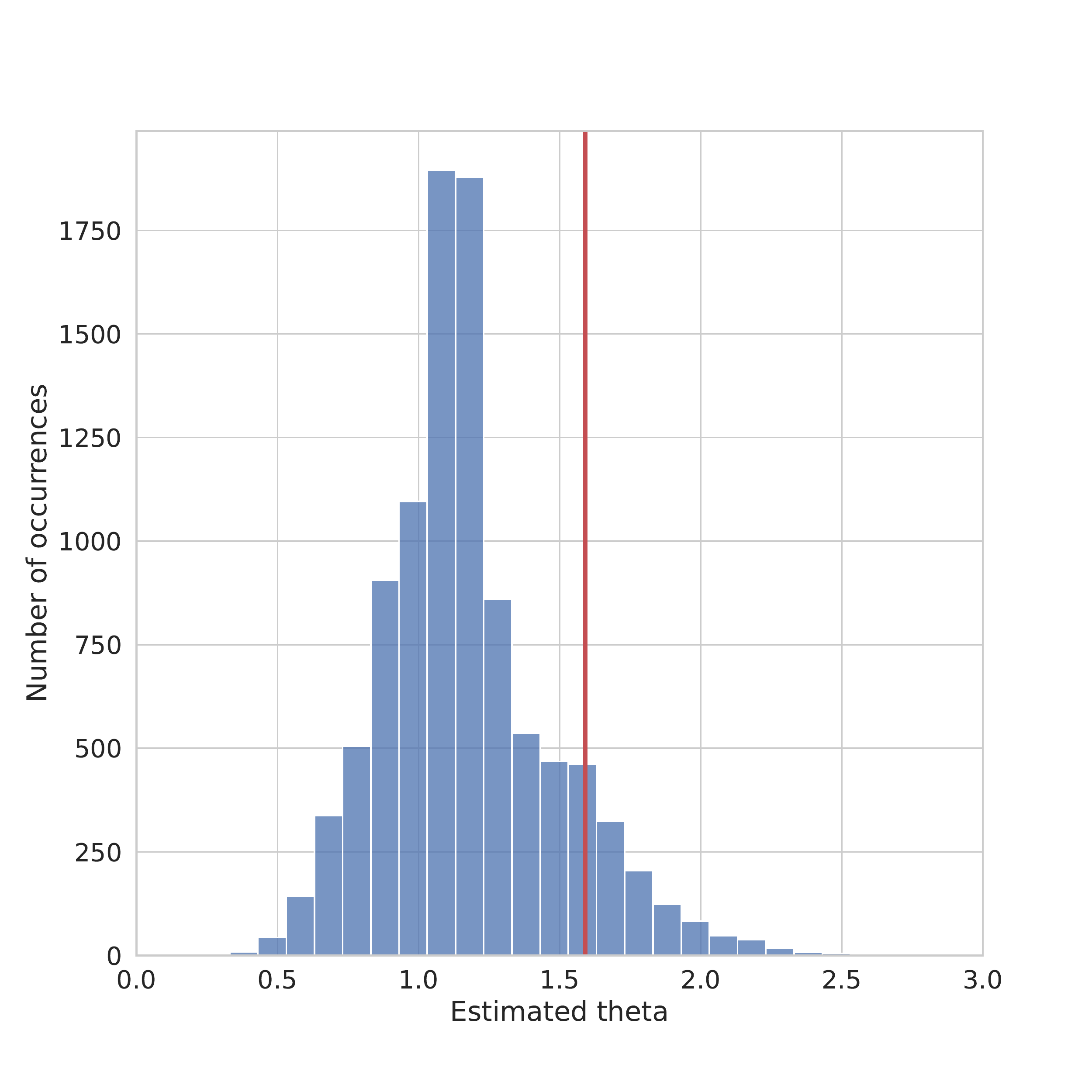}}}
\end{tabular}

\medskip
\centering
{\small\textit{Note}:  The figures are shown in the same format as Figure~\ref{fig:homogeneous}.}
\caption{The histograms of the estimates under heterogeneous iROAS, where the heterogeneous iROAS has an additive term that is proportional to the geo size.}
\label{fig:heterogeneous}
\end{figure*}

\clearpage
\newpage

\section{Additional figures}\label{sec:other_figs}
See Figures~\ref{fig:design}, \ref{fig:homogeneous}, and~\ref{fig:heterogeneous} for additional empirical results referred to in the paper.

\section{Additional evaluations}\label{sec:other_hetero}
In this section we report additional results based on an alternative model for heterogeneous $\theta_g$. We set $\theta_g = 1 + u_g$, where $u_g$ is an independent random value generated from the uniform distribution over $[-0.25, 0.25]$. The results are shown in Table~\ref{tab:heterogeneous_unif} and Figure~\ref{fig:heterogeneous_unif}.

\begin{table}[!h]
\centering
\begin{tabular}{l|l|cc|cc}
    \toprule
    \multirow{2}{*}{Est.} & \multirow{2}{*}{Design} & \multicolumn{2}{c|}{Dataset A} & \multicolumn{2}{c}{Dataset B} \\
    \cmidrule(r){3-6}
    & & RMSE & Bias & RMSE & Bias \\
    \midrule
    \multirow{2}{*}{$\hat{\theta}$} & Pairs & $0.96$ & $0.032$ & $2.04$ & $0.119$ \\
     & Supergeo & $0.41$ & $0.001$ & $0.29$ & $0.006$ \\
    \midrule
    \multirow{2}{*}{$\hat{\theta}^{\mathrm{trim}}$} & Pairs & $0.32$ & $0.002$ & $0.43$ & $0.057$ \\
     & Supergeo & $0.32$ & $0.012$ & $0.41$ & $0.025$ \\
    \bottomrule
  \end{tabular}
  
\medskip
\centering
{\small\textit{Note}: The table is shown in the same format as Table~\ref{tab:homogeneous}.}
\caption{Empirical results under heterogeneous iROAS, where $\theta_g$ is 1 plus a uniformly random noise.}
\label{tab:heterogeneous_unif}
\end{table}

\begin{figure*}[!ht]
\centering
\begin{tabular}{@{} c|c @{}}
\subfigure[A: $\hat{\theta}$ of Pairs]{\label{fig:A_hetero_unif_no_trim}{\includegraphics[width=0.24\textwidth]{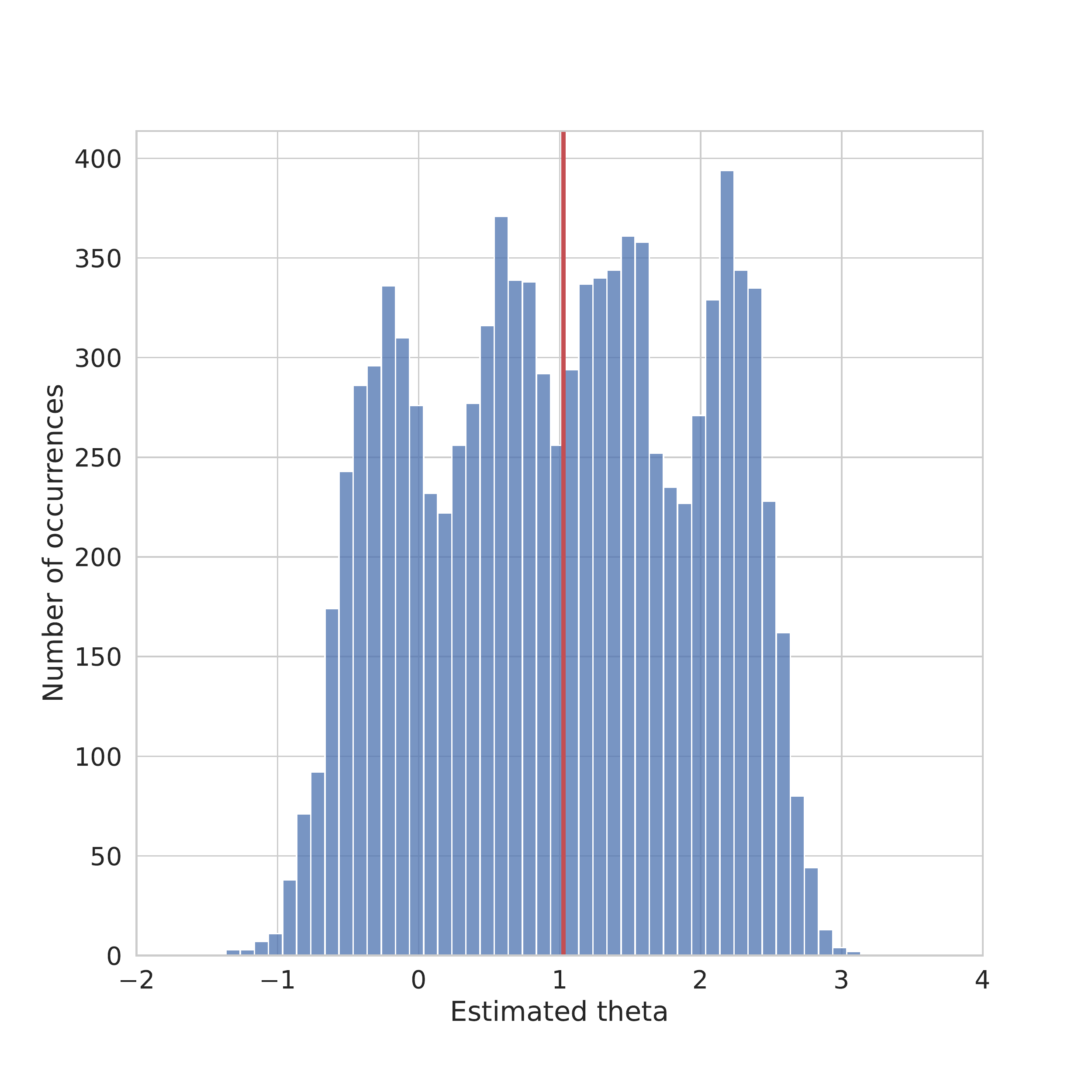}}}
\subfigure[A: $\hat{\theta}$ of Supergeo]{\label{fig:A_hetero_unif_supergeo}{\includegraphics[width=0.24\textwidth]{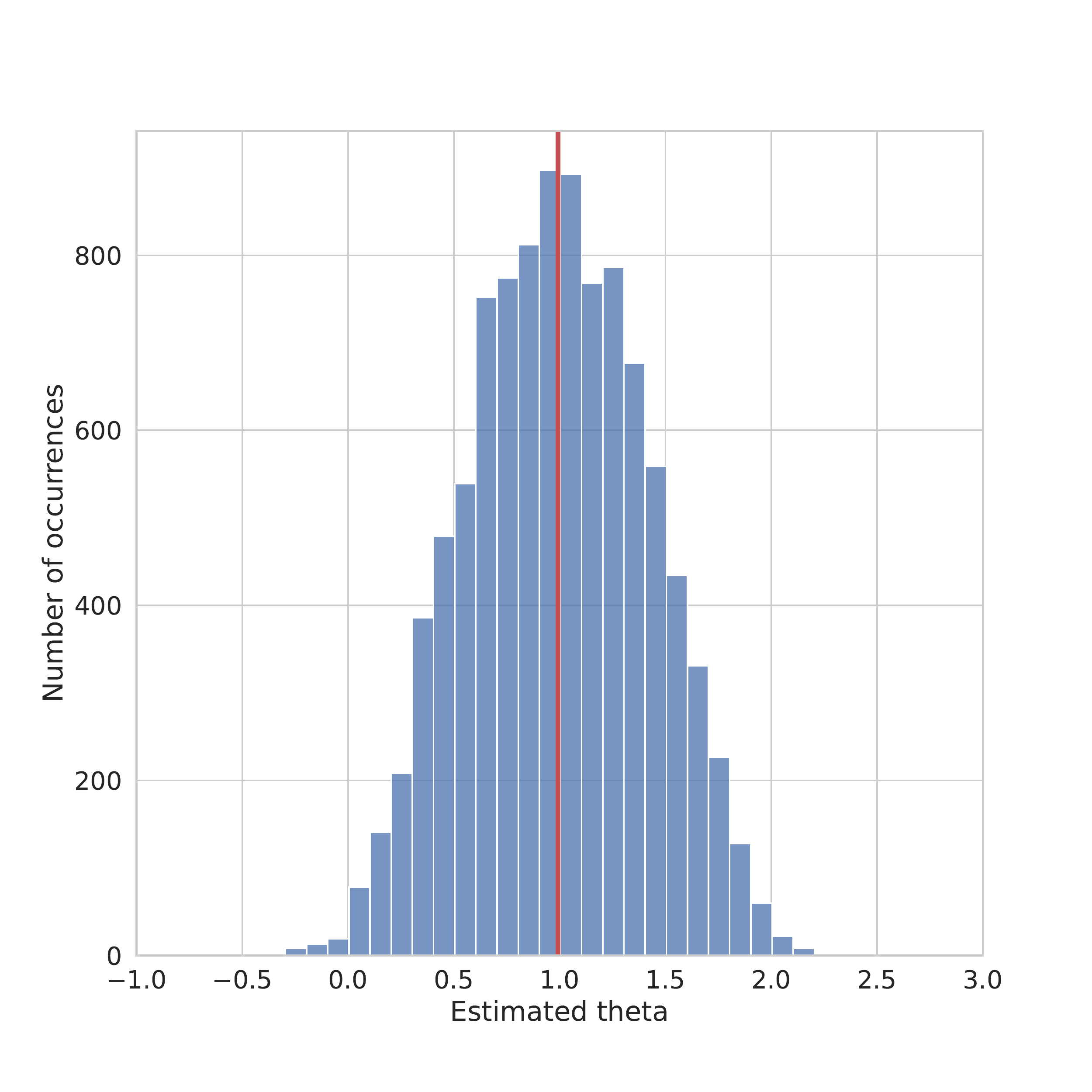}}}
&
\subfigure[A: $\hat{\theta}^{\mathrm{trim}}$ of Pairs ]{\label{fig:A_hetero_unif_with_trim}{\includegraphics[width=0.24\textwidth]{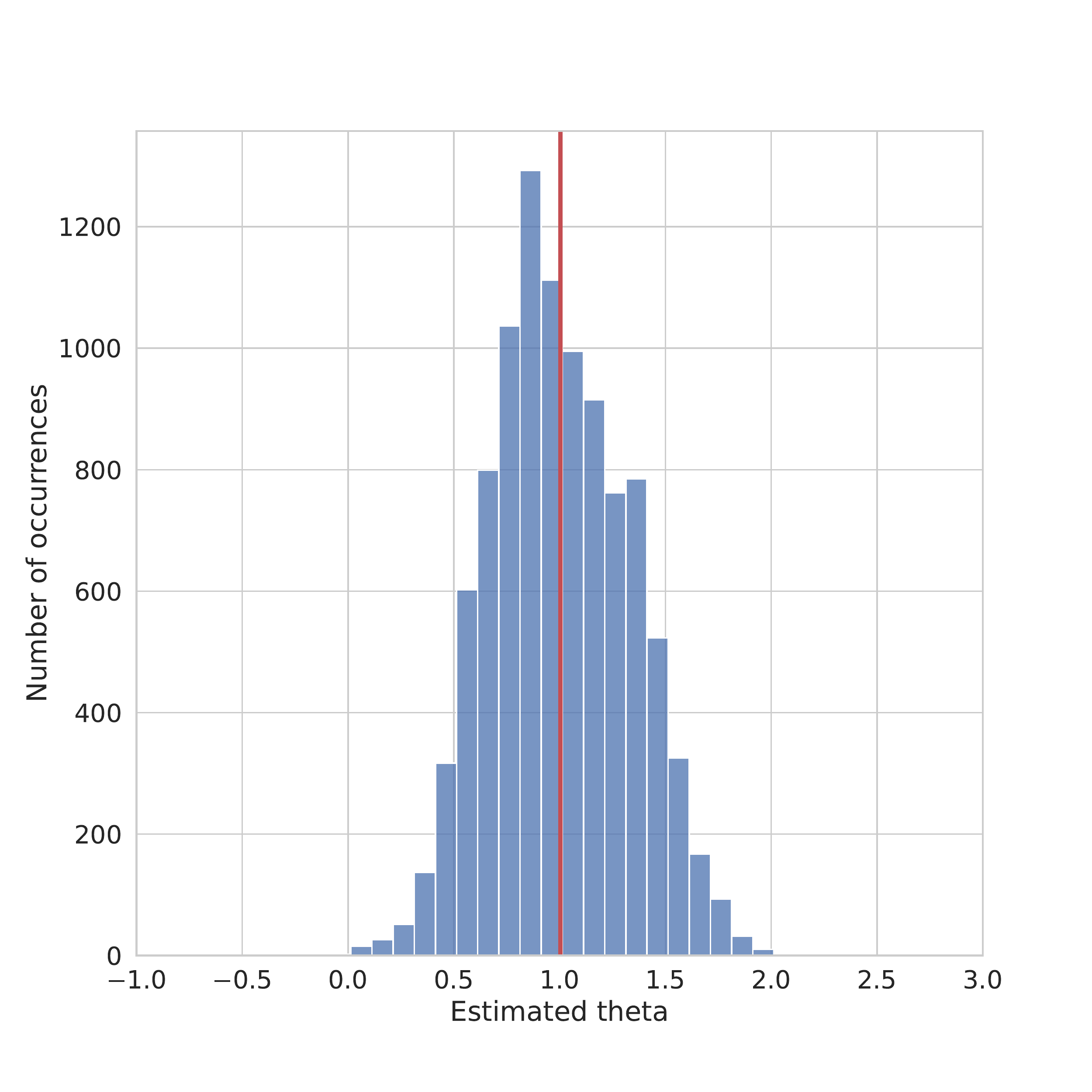}}}
\subfigure[A: $\hat{\theta}^{\mathrm{trim}}$ of Supergeo]{\label{fig:A_hetero_unif_supergeo_trim}{\includegraphics[width=0.24\textwidth]{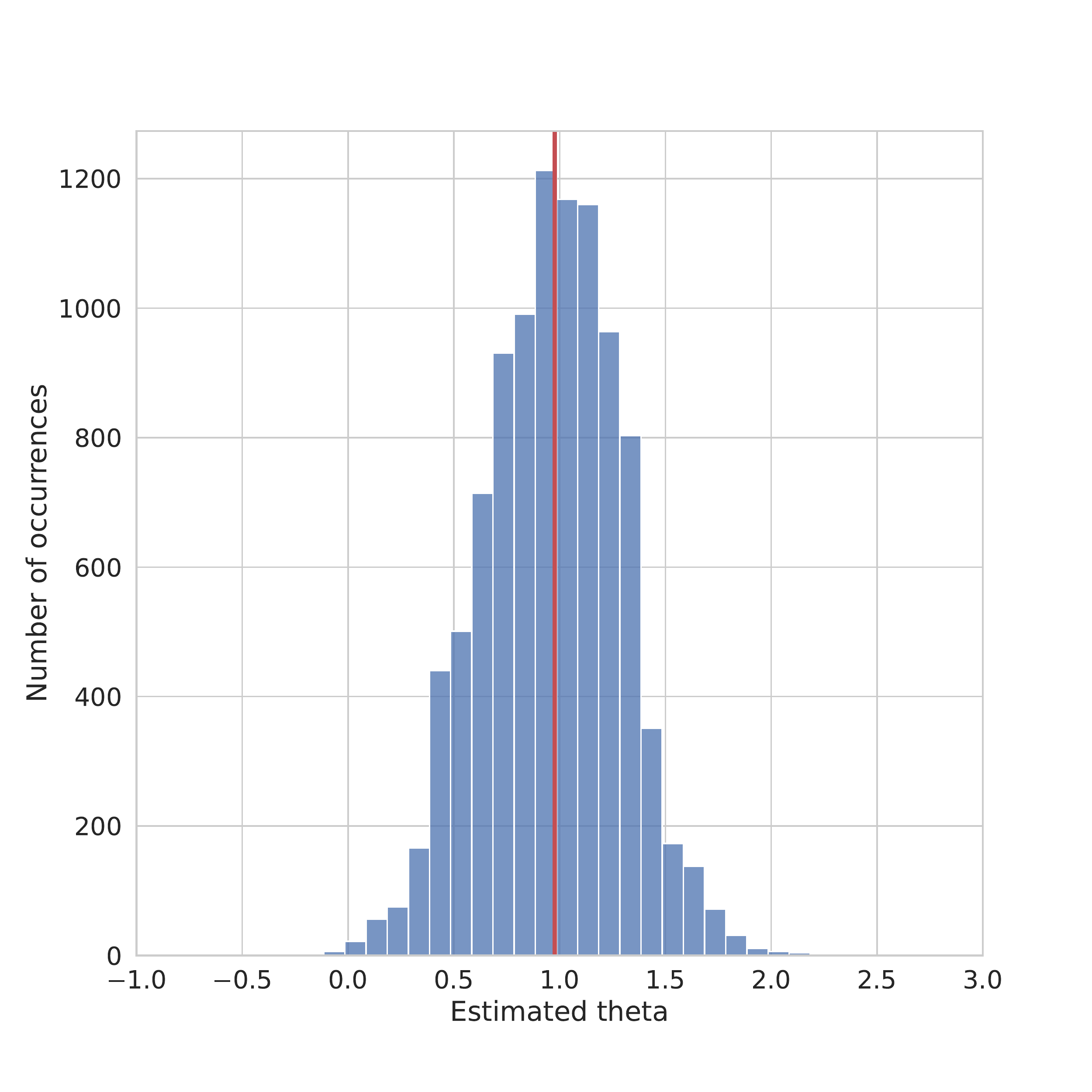}}}
\\ \hline
\subfigure[B: $\hat{\theta}$ of Pairs]{\label{fig:B_hetero_unif_no_trim}{\includegraphics[width=0.24\textwidth]{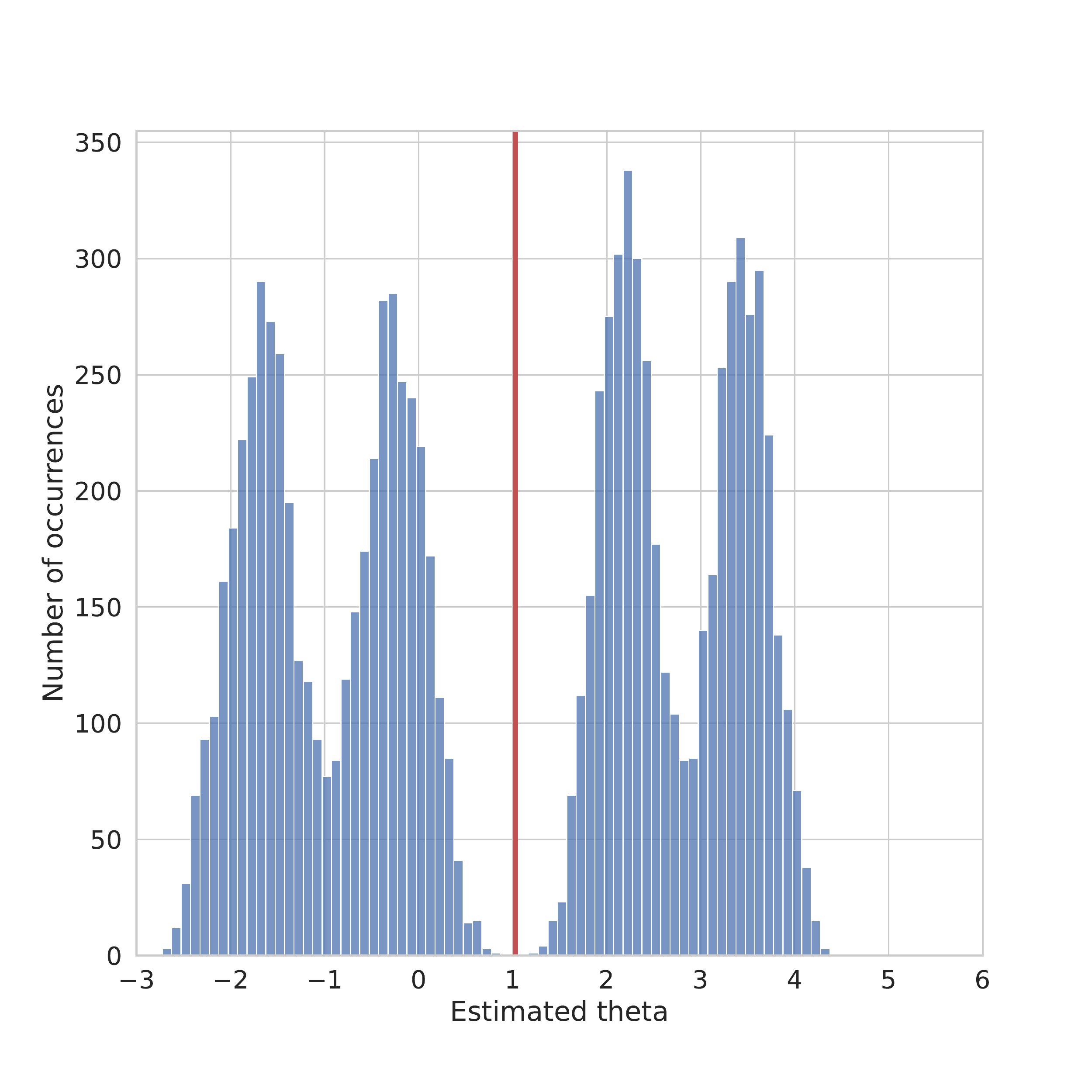}}}
\subfigure[B: $\hat{\theta}$ of Supergeo]{\label{fig:B_hetero_unif_supergeo}{\includegraphics[width=0.24\textwidth]{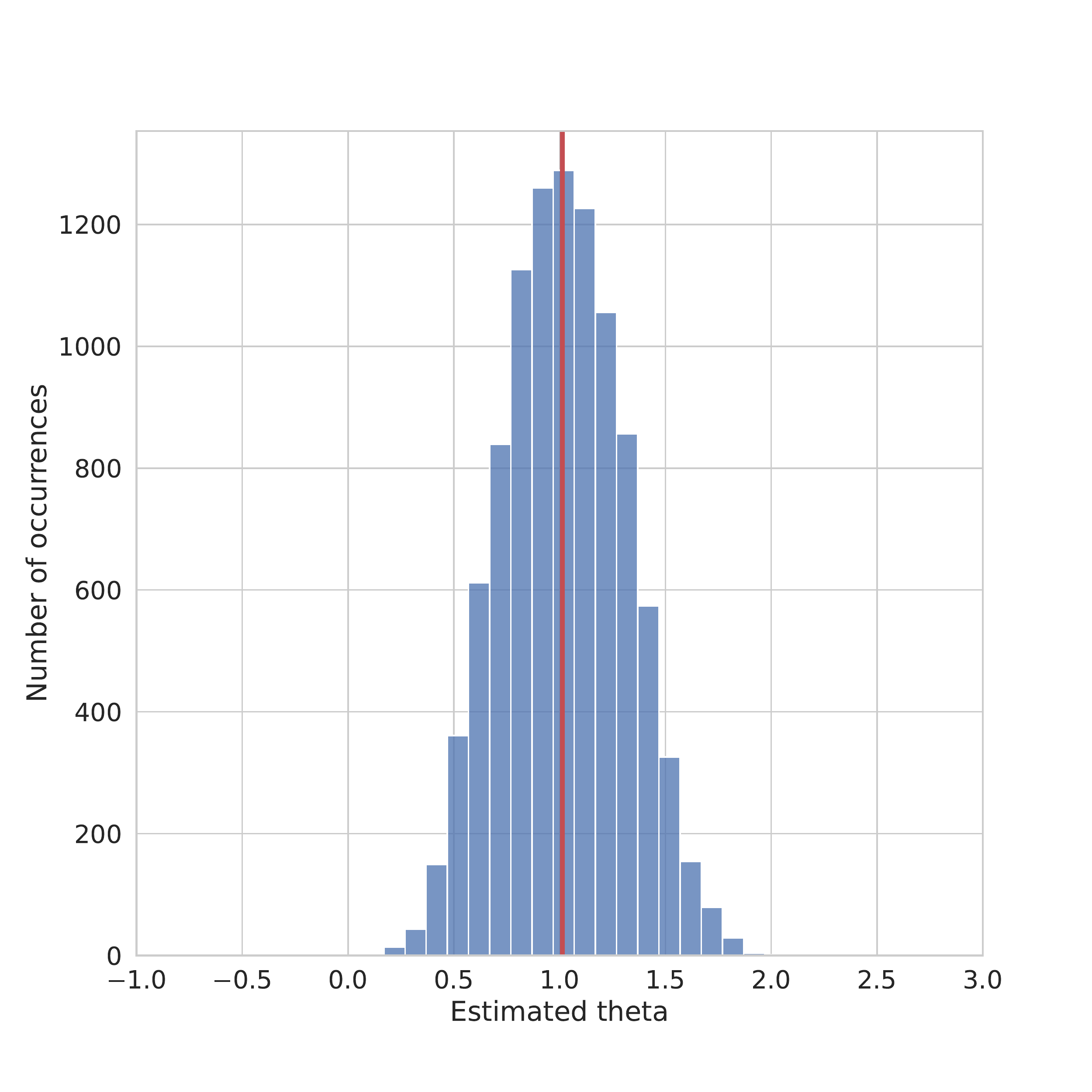}}}
&
\subfigure[B: $\hat{\theta}^{\mathrm{trim}}$ of Pairs ]{\label{fig:B_hetero_unif_with_trim}{\includegraphics[width=0.24\textwidth]{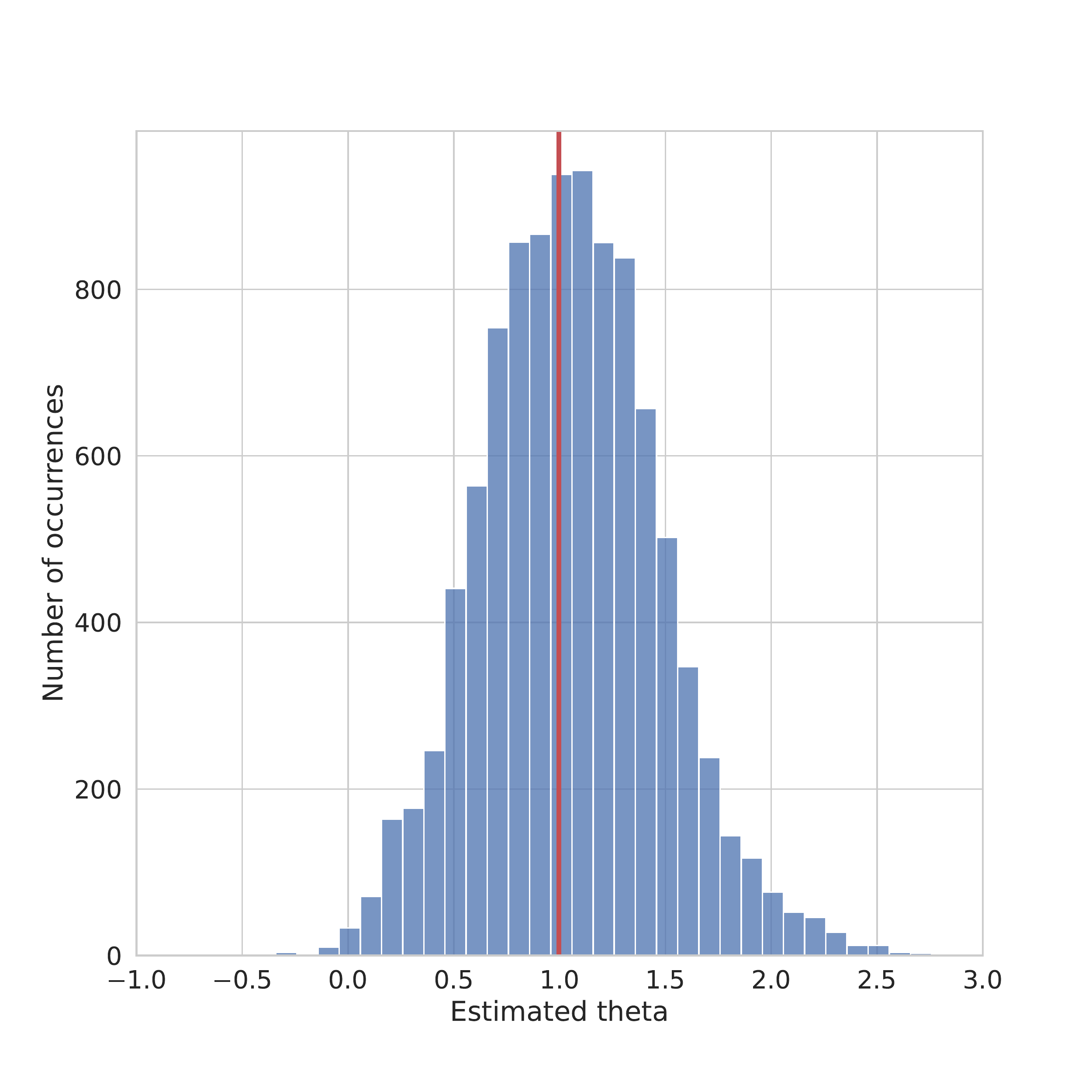}}}
\subfigure[B: $\hat{\theta}^{\mathrm{trim}}$ of Supergeo]{\label{fig:B_hetero_unif_supergeo_trim}{\includegraphics[width=0.24\textwidth]{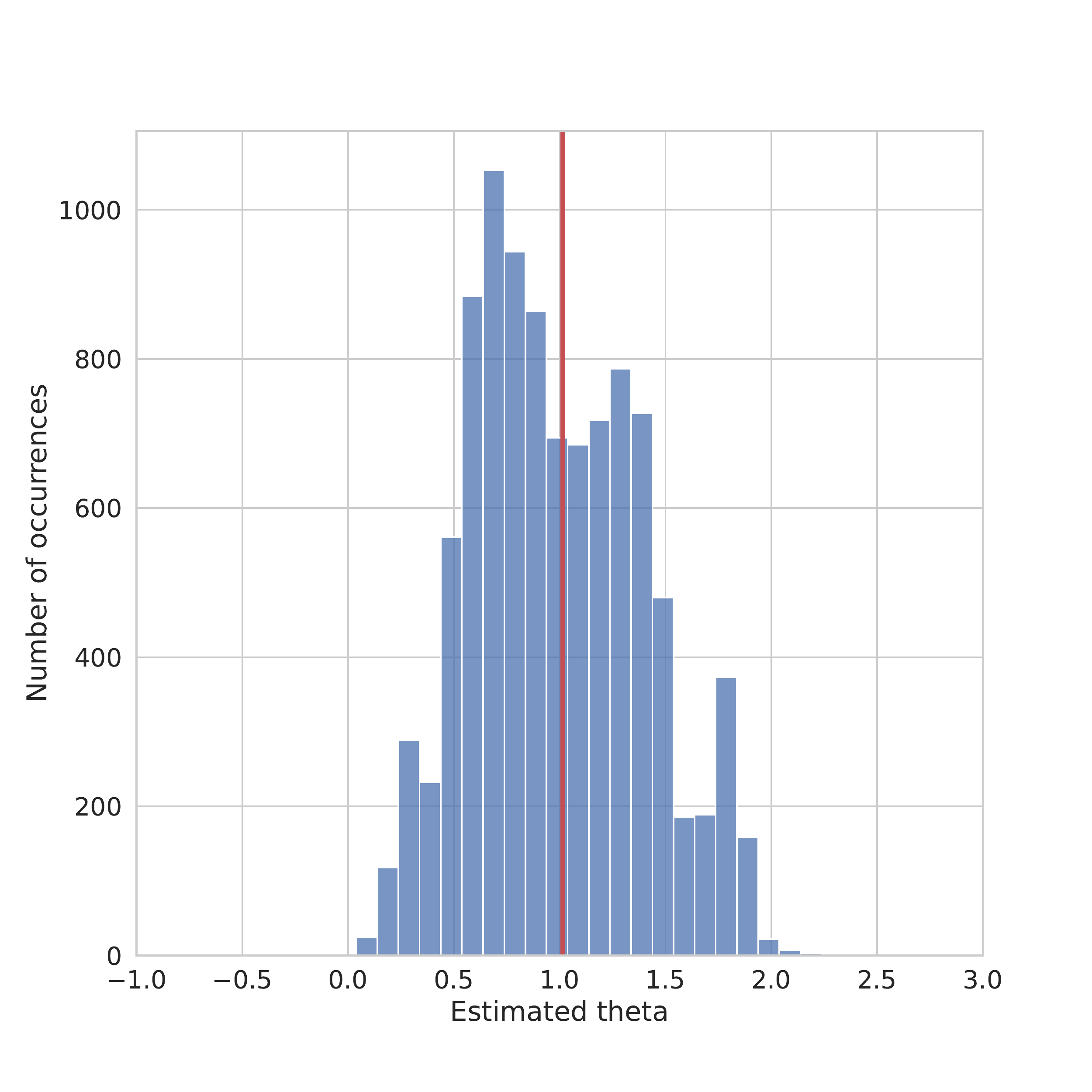}}}
\end{tabular}

\medskip
\centering
{\small\textit{Note}: The figures are shown in the same format as Figure~\ref{fig:homogeneous}.}
\caption{The histogram of the estimates under heterogeneous iROAS, where the heterogeneous iROAS is 1 plus a uniformly random noise.}
\label{fig:heterogeneous_unif}
\end{figure*}

\section{Choosing supergeo sizes}\label{sec:size}
In this section we present additional theoretical results that speak to the robustness properties of randomized designs. In the spirit of \citet{banerjee2020theory}, we want to capture a scenario in which an approach to designing an experiment is presented to a potentially adversarial audience. If the design is ``over-optimized''---for example, when there are two large supergeos and only two possible treatment assignments---the adversary can easily come up with a model that compromises the experimental results. For instance, they can suggest the presence of an unobserved confounder that differs between the two supergeos and that is correlated with the outcome variable. When the treatment assignment is realized as a result of many independent random draws---as is the case with many smaller supergeo pairs---the job of the adversary becomes substantially more complicated.

We also present additional empirical results that illustrate the benefits of using the maximum allowed supergeo pair size within the range of $\{3,4,5\}$.


\paragraph{Model.}
We follow the notation from Section~\ref{sec:homo_model} and maintain the fixed budget and the homogeneous $\theta$ assumptions. We modify the linear outcome model in Assumption~\ref{ass:linear} to include an adversarial noise $\delta_g$:
\[
R_g = \theta \cdot S_g + Z_g + \delta_g.
\]
\begin{itemize}
    \item The adversary picks the values of per-geo noises $\delta_g$ after observing the allocation of geos to supergeo pairs $\{(G_{k,+}, G_{k,-})\}_{k=1}^K$, but \emph{before} the realization of a particular random treatment assignment.
    \item The noise $\delta_g$ is bounded by $|\delta_g| \leq \eta \cdot Z_g$ for some constant $\eta$.
\end{itemize}

\paragraph{Expectation.}
Following an argument similar to that of Section~\ref{sec:homo_model} and defining $\delta_G := \sum_{g \in G} \delta_g$ for any $G \subseteq \mathcal{G}$, we have
\begin{align*}
    \hat{\theta} = &~ \frac{\sum_{g \in \mathcal{T}} R_g - \sum_{g' \in \mathcal{C}} R_{g'}}{\sum_{g \in \mathcal{T}} S_g - \sum_{g' \in \mathcal{C}} S_{g'}} \notag \\
    = &~ \theta + \frac{\sum_{g \in \mathcal{T}} (Z_g + \delta_g) - \sum_{g' \in \mathcal{C}} (Z_{g'} + \delta_{g'})}{\sum_{g \in \mathcal{T}} S_g - \sum_{g' \in \mathcal{C}} S_{g'}} \notag \\
    = &~ \theta + \frac{1}{B} \cdot \sum_{k} A_{k} \cdot (Z_{G_{k,+}} + \delta_{G_{k,+}} - Z_{G_{k,-}} - \delta_{G_{k,-}}).
\end{align*}
Since the random draws $A_{k}$'s are independent of the noises, $\delta_g$'s, the empirical estimator, $\hat{\theta}$, is still unbiased regardless of the values of $\delta_g$'s.

\paragraph{Variance.}
For any supergeo design $\{(G_{k,+}, G_{k,-})\}_{k=1}^K$, estimator $\hat{\theta}$ is a function of $\delta_g$'s and its variance is:
\begin{align*}
    var[\hat{\theta}(\{\delta_g\})] = &~ \frac{1}{B^2} \cdot \E\Big[ \Big( \sum_{k} A_{k} \cdot (Z_{G_{k,+}} + \delta_{G_{k,+}} - Z_{G_{k,-}} - \delta_{G_{k,-}}) \Big)^2 \Big] \\
    = &~ \frac{1}{B^2} \cdot \sum_{k} (Z_{G_{k,+}} + \delta_{G_{k,+}} - Z_{G_{k,-}} - \delta_{G_{k,-}})^2,
\end{align*}
where the second equality follows from that the fact that random draws $A_{k}$'s are independent of the noises $\delta_g$'s as well as independent from each other.

Let's compute the largest possible variance that can be achieved by the adversary picking the values of $\delta_g$'s. Consider any supergeo pair $(G_{k,+}, {G_{k,-}})$, and assume w.l.o.g.~that $Z_{G_{k,+}} \geq Z_{G_{k,-}}$. The largest value of the variance is realized when the adversary picks $\delta_{G_{k,+}} = \eta \cdot Z_{G_{k,+}}$ and $\delta_{G_{k,-}} = - \eta \cdot Z_{G_{k,-}}$:
\begin{align}\label{eq:variance_adversarial}
    \max_{\{\delta_g\}}\ var[\hat{\theta}(\{\delta_g\})] = &~ \max_{\{\delta_g\}}\ \frac{1}{B^2} \cdot \sum_{k} (Z_{G_{k,+}} + \delta_{G_{k,+}} - Z_{G_{k,-}} - \delta_{G_{k,-}})^2 \notag \\
    = &~ \frac{1}{B^2} \cdot \sum_{k} ( |Z_{G_{k,+}} - Z_{G_{k,-}}| + \eta \cdot Z_{G_{k,+}} + \eta \cdot Z_{G_{k,-}})^2 \notag \\ 
    = &~ \frac{1}{B^2} \cdot \sum_{k} \Big( (Z_{G_{k,+}} - Z_{G_{k,-}})^2 + 2 \eta \cdot |Z_{G_{k,+}}^2 - Z_{G_{k,-}}^2| + \eta^2 \cdot (Z_{G_{k,+}} + Z_{G_{k,-}})^2 \Big).
\end{align}
Importantly, the third term, $\eta^2 \cdot (Z_{G_{k,+}} + Z_{G_{k,-}})^2$, penalizes supergeo pairs that are too large. This motivates the decision to construct many smaller supergeo pairs.

\subsection{Empirical results}\label{sec:appendix-B-eval}
In this section we compare supergeo designs with different maximum allowed sizes of supergeo pairs. As discussed in Section~\ref{sec:supergeo_algo}, finding the optimal supergeo design can be computationally expensive especially when the subset size is large. To mitigate this issue, we work with a synthetic dataset consisting of 40 geos. We generate the synthetic data in the same way as in Section~5.1 of \cite{chen2021trimmed}.  

\paragraph{Pretest loss vs.~test loss.}
Figure~\ref{fig:sync_40_loss} shows the loss from Eq.~\eqref{eq:loss} as a function of the maximum allowed size of a supergeo pair which varies from 2 to 8. Recall that the loss is proportional to the variance of the estimated $\hat{\theta}$. When computing the loss, we approximate $Z_g$ either by the pretest response $R_g^{\mathrm{pre}}$ or the test response $R_g^{\mathrm{test}}$. Even though a larger subset size is effective at decreasing the pretest loss (i.e., the training loss), the test loss reaches the minimum when the subset size is 4 and remains roughly the same when we increase it further. This confirms our intuition that a large subset size leads to overfitting to the pretest data.
\begin{figure}[!h]
\centering
\includegraphics[width=0.5\textwidth]{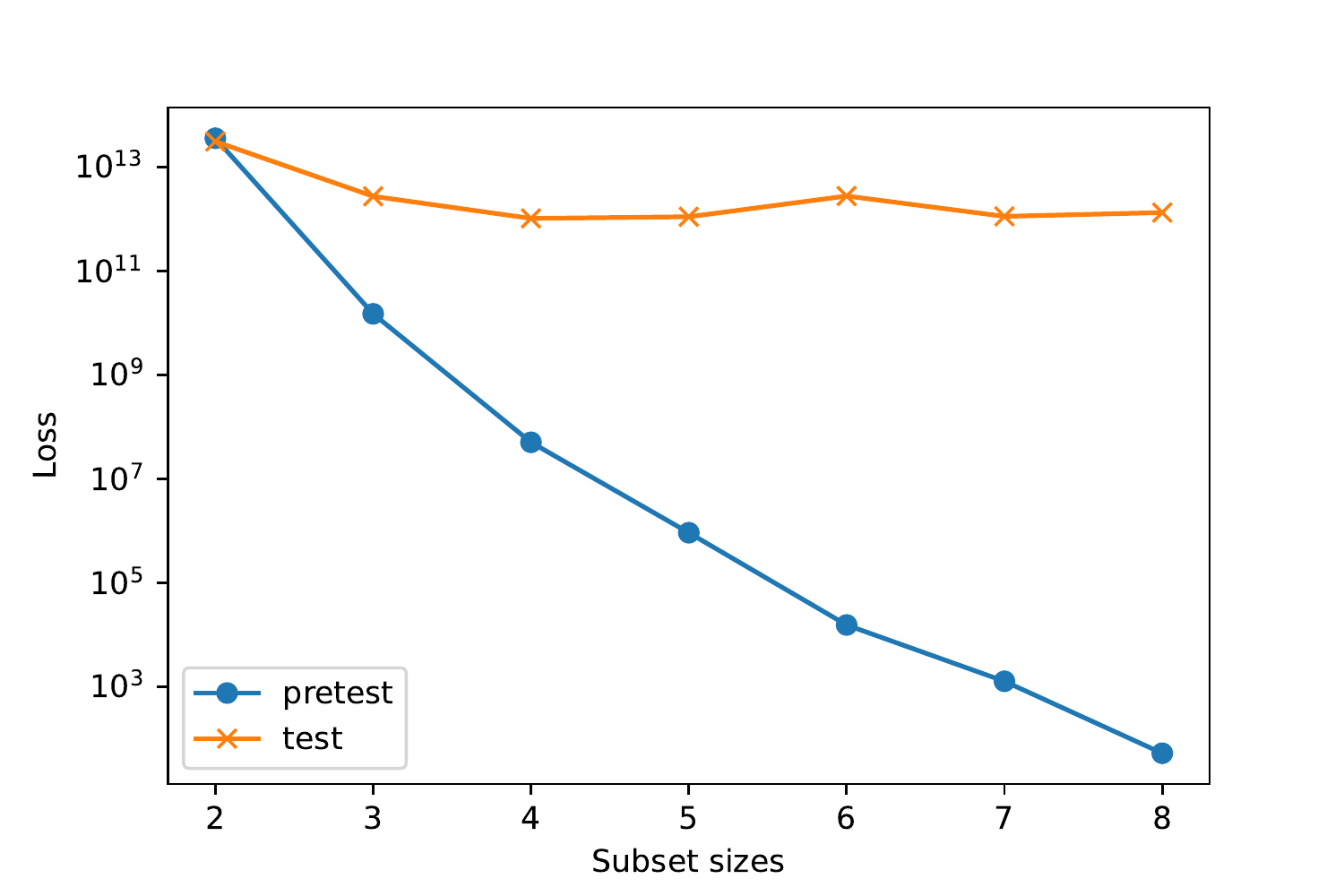}

\medskip
\raggedright
{\small\textit{Note}: Reporting the loss on the pretest data, $\sum_{k=1}^K (R_{G_{k,+}}^{\mathrm{pre}} - R_{G_{k,-}}^{\mathrm{pre}})^2$, and the loss on the test data $\sum_{k=1}^K (R_{G_{k,+}}^{\mathrm{test}} - R_{G_{k,-}}^{\mathrm{test}})^2$, as the subset size varies from 2 to 8. The loss is shown using a log scale.}
\caption{Loss on the pretest and test data.}
\label{fig:sync_40_loss}
\end{figure}

\paragraph{No adversarial noise.} We repeat the evaluations from Section~\ref{sec:evaluation_homogeneous} with $\theta_g = 1$ for all $g$ on the synthetic dataset with 40 geos. The results are shown in Table~\ref{tab:sync_40_geo}. We observe that the RMSE reaches the minimum when the subset size is 4 or 3 (depending on the estimator used) and does not decrease further when the subset size grows. 
\begin{table}[!h]
\centering
\begin{tabular}{l|c|cc|cc}
    \toprule
    & \multirow{2}{*}{\# pairs} & \multicolumn{2}{c|}{Empirical ($\hat{\theta}$)} & \multicolumn{2}{c}{Trimmed ($\hat{\theta}^{\mathrm{trim}}$)} \\
    \cmidrule(r){3-6}
    & & RMSE & Cov. & RMSE & Cov. \\
    \midrule
    Size 2 & 20 & $0.554$ & $76\%$ & $0.326$ & $70\%$ \\
    Size 3 & 14 & $0.163$ & $85\%$ & $\mathbf{0.059}$ & $76\%$ \\
    Size 4 & 10 & $\mathbf{0.100}$ & $76\%$ & $0.099$ & $62\%$ \\
    Size 5 & 8 & $0.104$ & $76\%$ & $0.105$ & $67\%$ \\
    Size 6 & 7 & $0.165$ & $75\%$ & $0.139$ & $63\%$ \\
    Size 7 & 6 & $0.105$ & $78\%$ & $0.114$ & $75\%$ \\
    Size 8 & 5 & $0.113$ & $69\%$ & $0.120$ & $63\%$ \\
    \bottomrule
  \end{tabular}
  
\medskip
\raggedright
{\small\textit{Note}: The coverage values (abbreviated as Cov.) are shown for nominally 80\% confidence intervals for supergeo designs with subset sizes from 2 to 8 on a synthetic dataset with 40 geos. We show the results of both the empirical estimator $\hat{\theta}$ and the trimmed match estimator $\hat{\theta}^{\mathrm{trim}}$. The confidence intervals are computed using the approximation by Student's $t$-distribution described in Section~5.2 of \citet{chen2022robust} (see Appendix~\ref{sec:inference}). We do not report the bias as it is close to zero (as is usually the case with homogeneous $\theta_g$'s).}
\caption{Root-mean-square errors (RMSE) of the estimates and the coverage.}
\label{tab:sync_40_geo}
\end{table}

\paragraph{With adversarial noise.}
We repeat the analysis with the adversarial noise added to the outcomes. Given a supergeo design $\{G_{k,+}, G_{k,-}\}_{k=1}^K$ and w.l.o.g.~assuming that $Z_{G_{k,+}} \geq Z_{G_{k,-}}$ for all $k$, we change $Z_{G_{k,+}}$ to $(1 + \eta) \cdot Z_{G_{k,+}}$ and change $Z_{G_{k,-}}$ to $(1 - \eta) \cdot Z_{G_{k,-}}$. We use $\eta = 0.05$ or $0.07$ in the evaluations.

The RMSEs of the empirical estimator are reported in Table~\ref{tab:sync_40_geo_adversarial}. We see that at first the RMSE declines as we increase the subset size which is due to the first two terms, $\sum_k (Z_{G_{k,+}} - Z_{G_{k,-}})^2 + 2 \eta \sum_k |Z_{G_{k,+}}^2 - Z_{G_{k,-}}^2|$, in the variance in Eq.~\eqref{eq:variance_adversarial}---the variance is lower when the pairs are better matched. However, as we further increase the subset size, the RMSE begins to grow due to the third term, $\eta^2 \cdot \sum_k (Z_{G_{k,+}} + Z_{G_{k,-}})^2$, which increases when the number of pairs decreases. In Figure~\ref{fig:sync_40_geo_adversarial} we include the histograms of the estimates $\hat{\theta}$ when $\eta = 0.07$.

\begin{table}[!h]
\centering
\begin{tabular}{l|c|c}
    \toprule
    & $\eta = 0.05$ & $\eta = 0.07$ \\
    \midrule
    Size 2 & $0.79$ & $0.90$ \\
    Size 3 & $0.44$ & $0.55$ \\
    Size 4 & $\mathbf{0.39}$ & $\mathbf{0.52}$ \\
    Size 5 & $0.42$ & $0.55$ \\
    Size 6 & $0.50$ & $0.63$ \\
    Size 7 & $0.47$ & $0.61$ \\
    Size 8 & $0.50$ &$0.65$ \\
    \bottomrule
  \end{tabular}
  
\medskip
\raggedright
{\small\textit{Note}: Reporting the RMSE of the empirical estimator, $\hat{\theta}$, in the adversarial perturbation model where $\eta = 0.05$ or $0.07$. We do not report the bias as it is close to zero (as is usually the case with homogeneous $\theta_g$'s).}
\caption{Root-mean-square errors (RMSE) of the empirical estimator.
}
\label{tab:sync_40_geo_adversarial}
\end{table}


\begin{figure*}[!ht]
\centering
\subfigure[Size 2]{\label{fig:size_2_adv}{\includegraphics[width=0.23\textwidth]{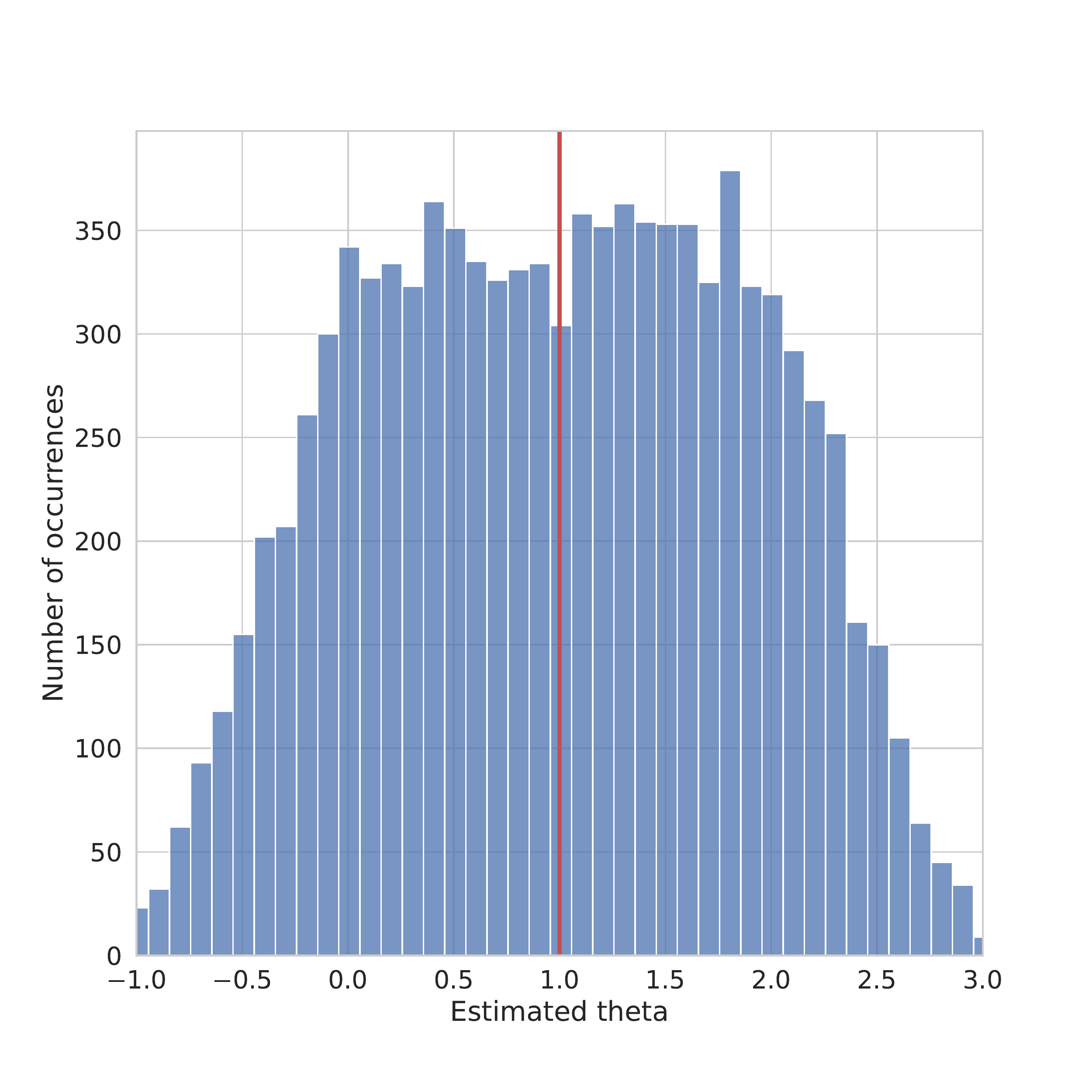}}}
\subfigure[Size 3]{\label{fig:size_3_adv}{\includegraphics[width=0.23\textwidth]{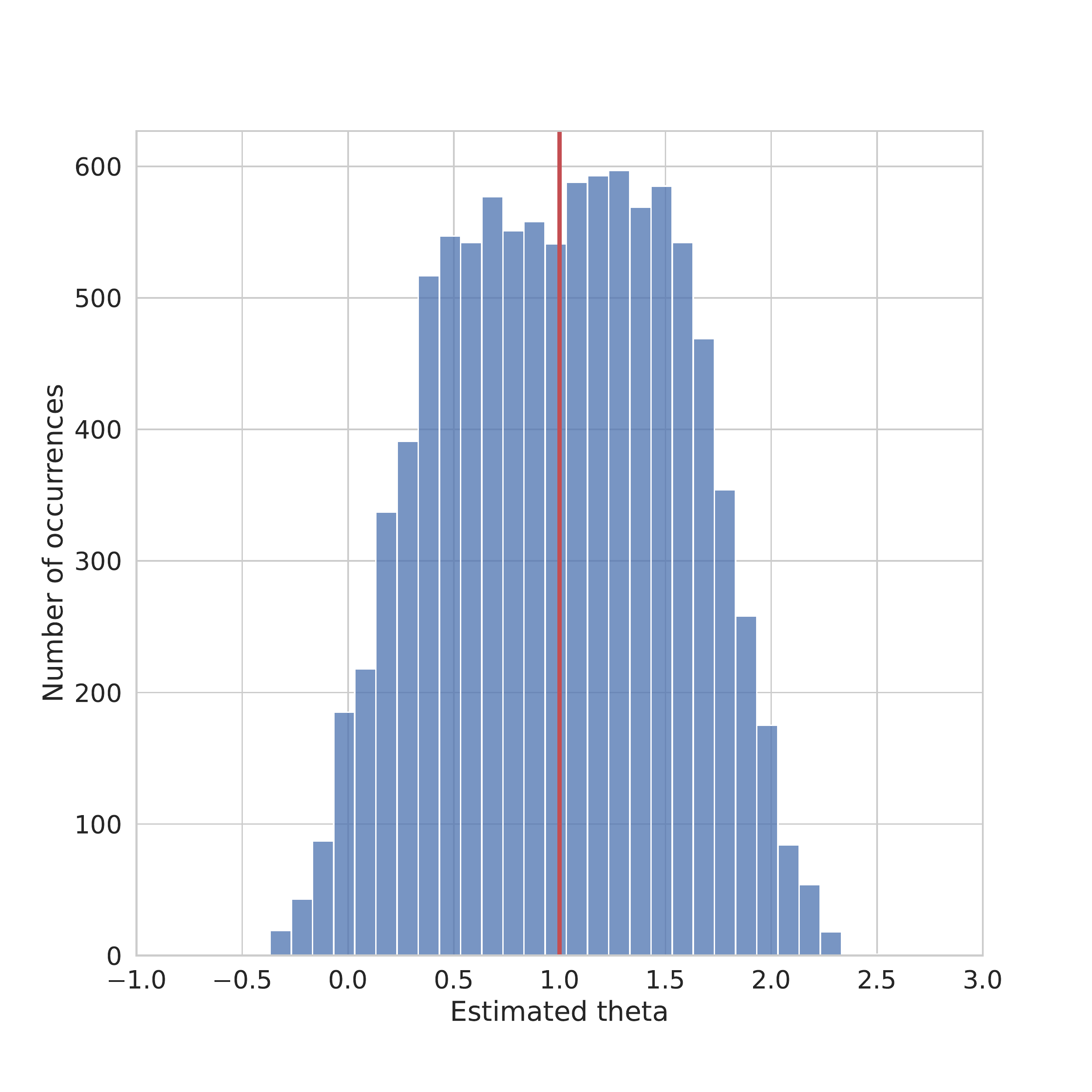}}}
\subfigure[Size 4]{\label{fig:size_4_adv}{\includegraphics[width=0.23\textwidth]{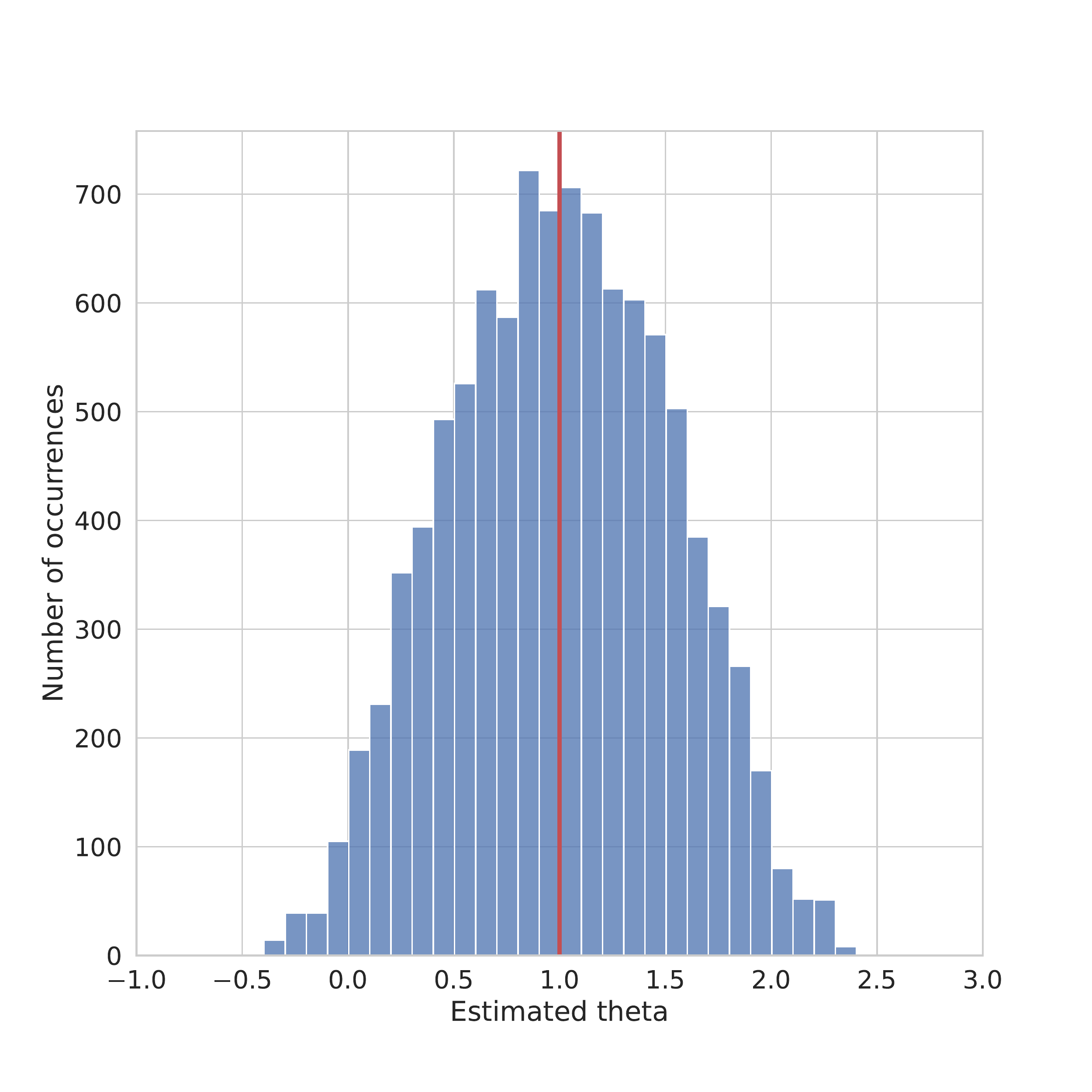}}}
\subfigure[Size 5]{\label{fig:size_5_adv}{\includegraphics[width=0.23\textwidth]{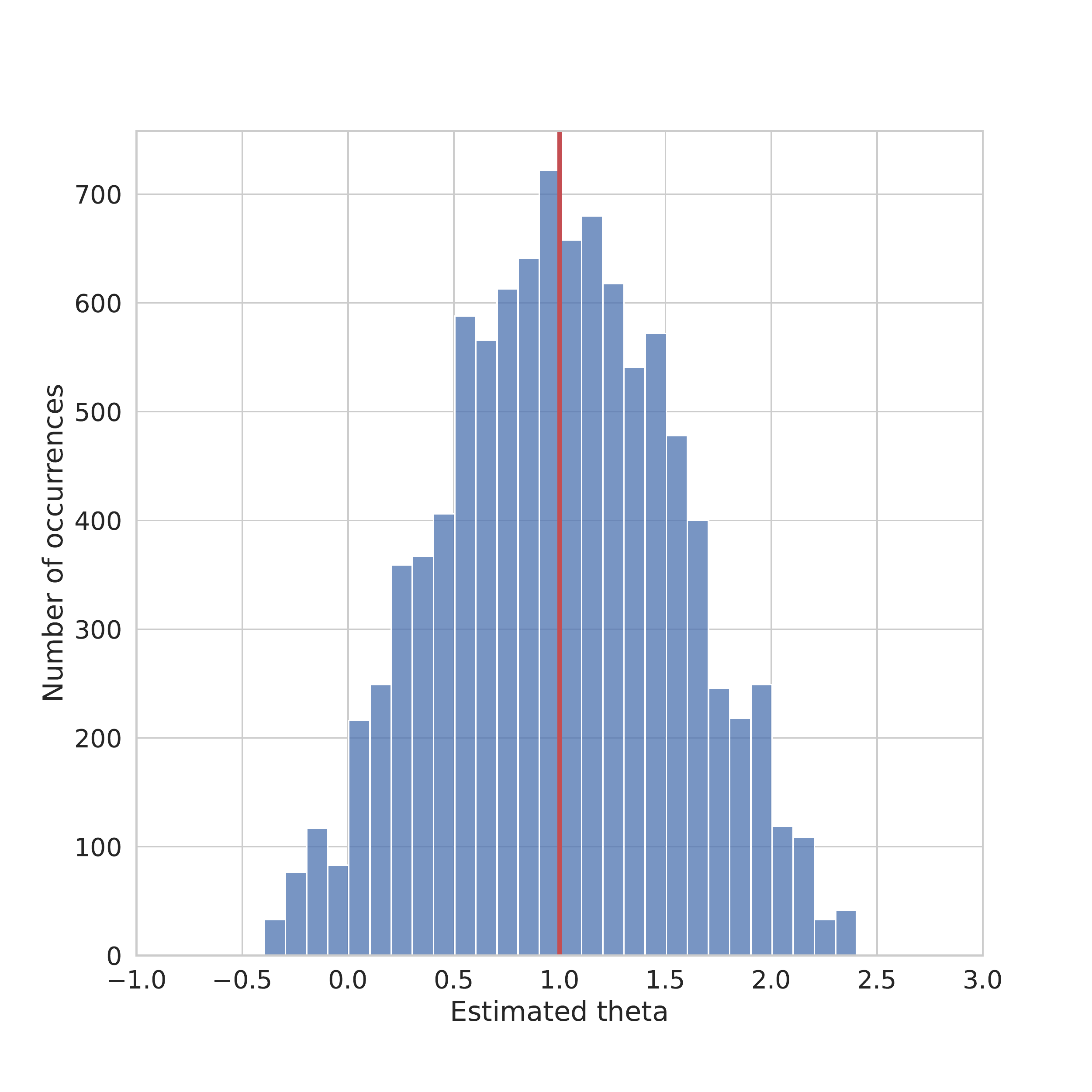}}}
\subfigure[Size 6]{\label{fig:size_6_adv}{\includegraphics[width=0.23\textwidth]{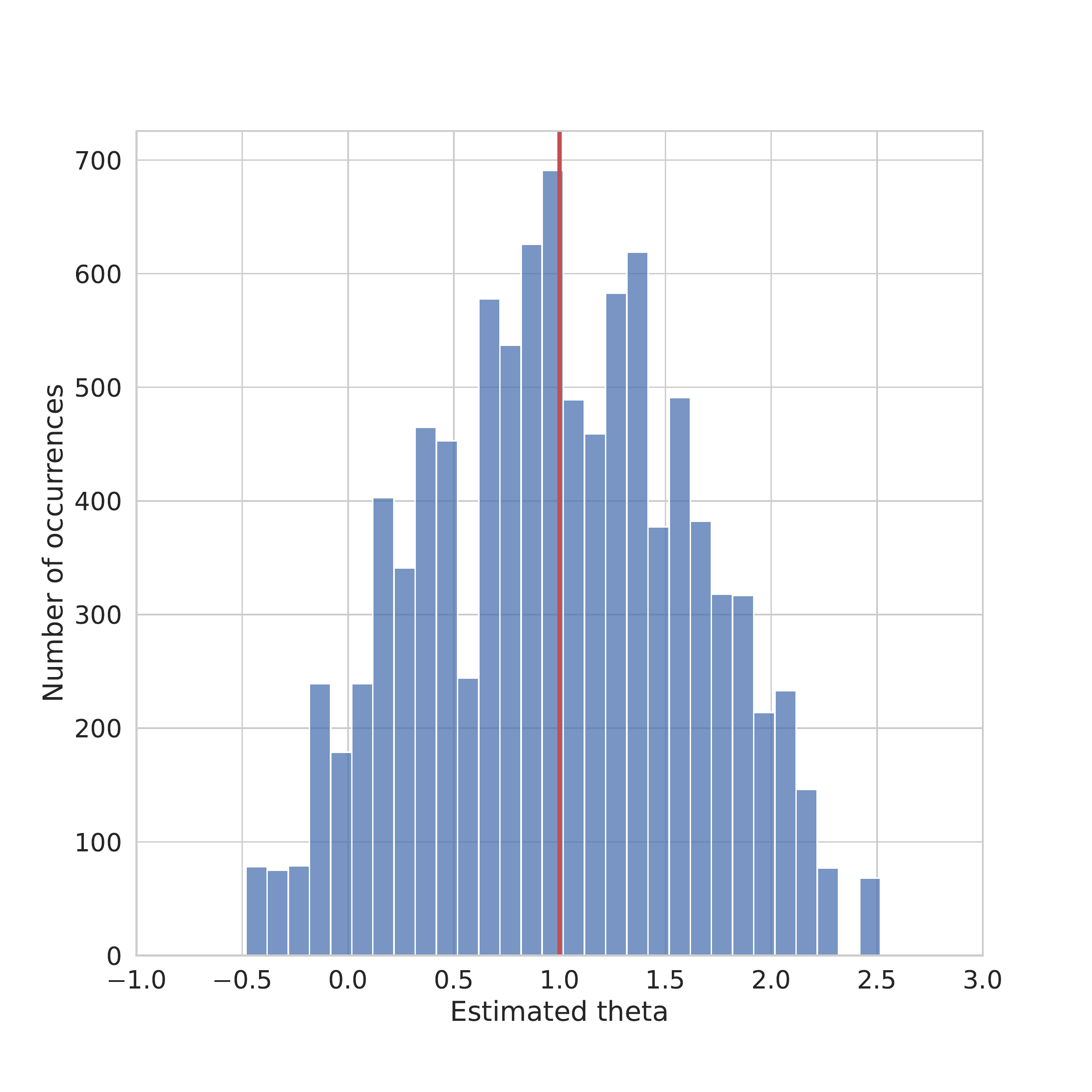}}}
\subfigure[Size 7]{\label{fig:size_7_adv}{\includegraphics[width=0.23\textwidth]{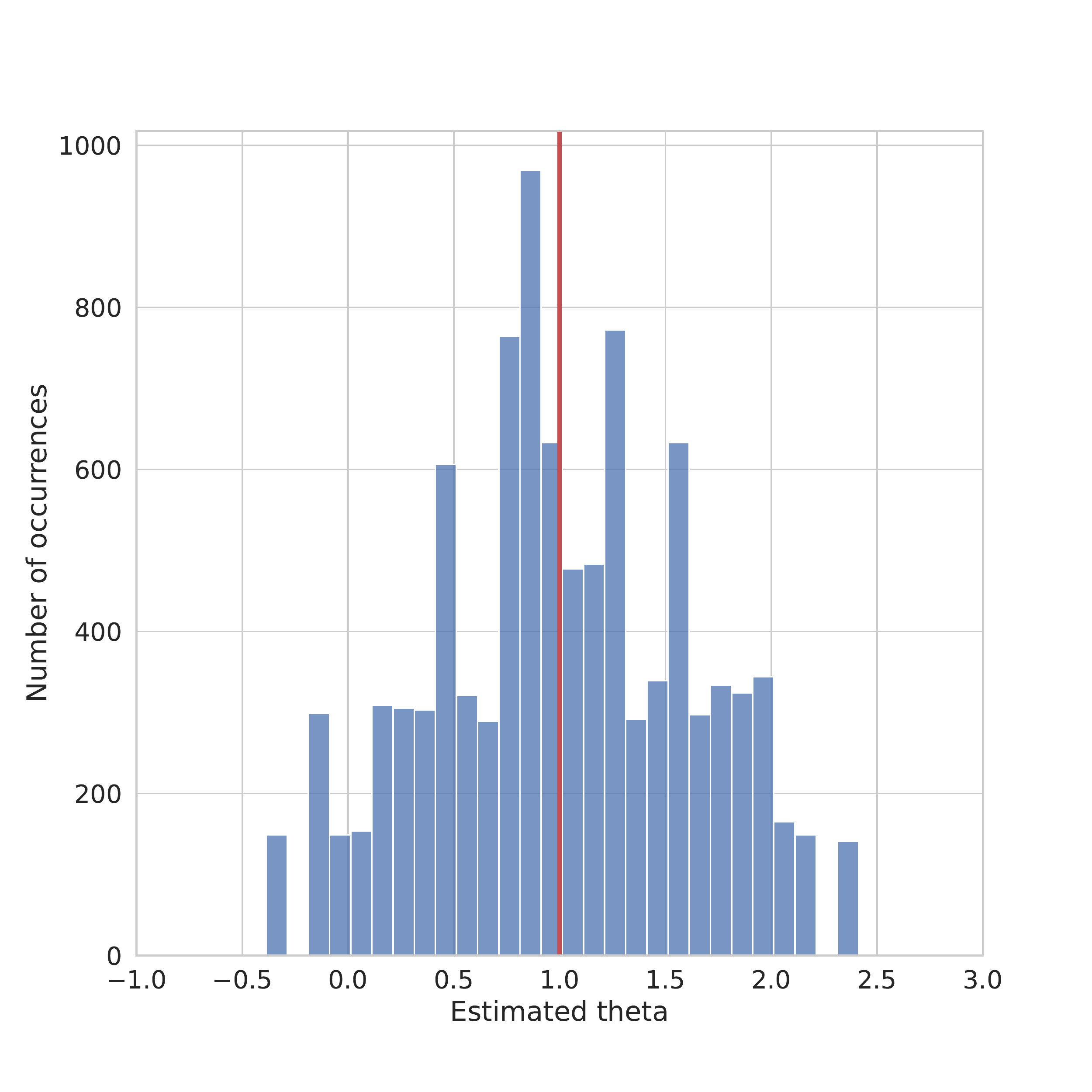}}}
\subfigure[Size 8]{\label{fig:size_8_adv}{\includegraphics[width=0.23\textwidth]{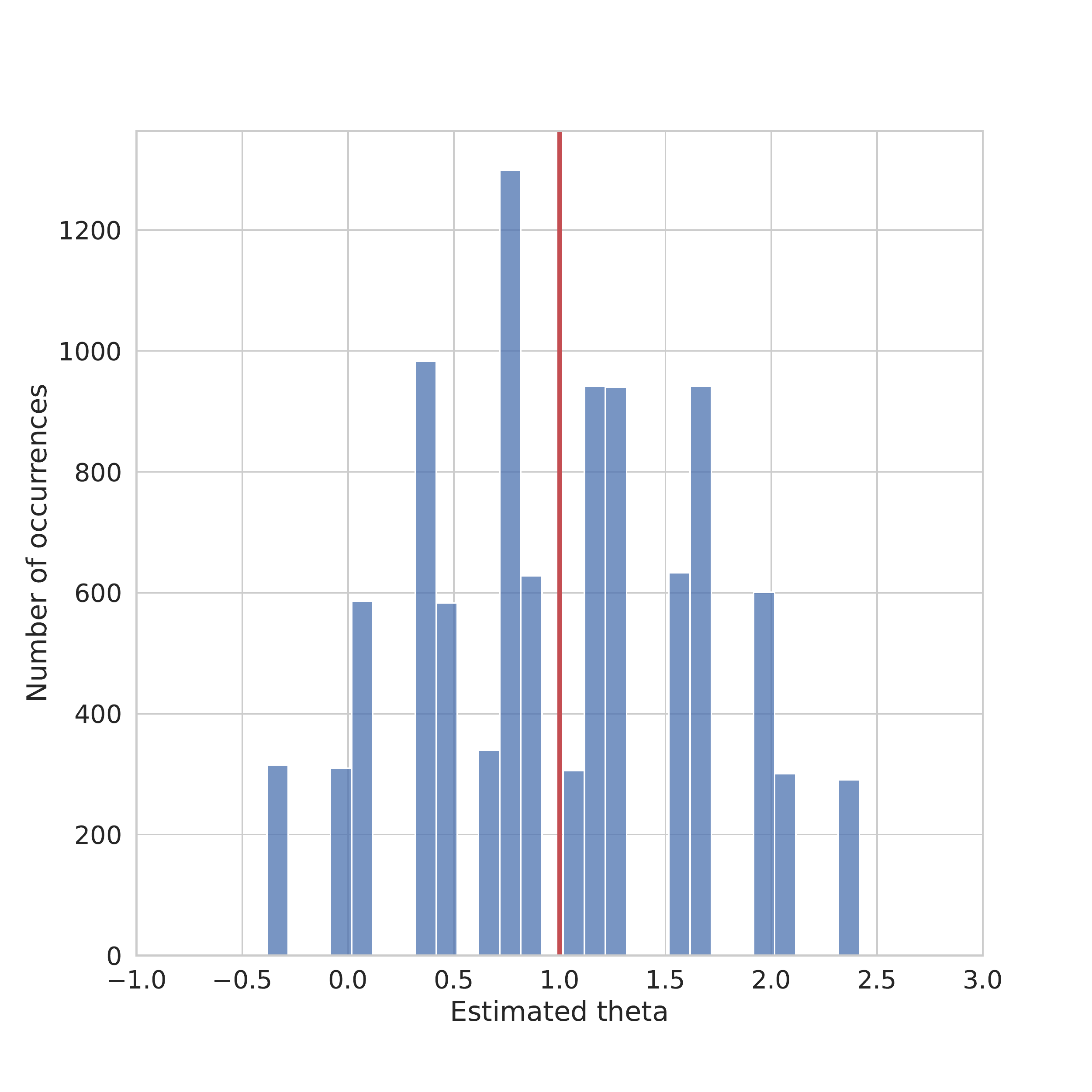}}}

\medskip
\centering
{\small\textit{Note}: The value of the perturbation parameter is $\eta = 0.07$.}
\caption{The histogram of the estimates $\hat{\theta}$ under adversarial perturbation.}
\label{fig:sync_40_geo_adversarial}
\end{figure*}
\section{Inference}\label{sec:inference}
This paper does not develop any novel approaches to statistical inference in matched pairs design. However, in this section we outline two ways that could be used to test hypotheses and construct confidence intervals.

\subsection{Approximation by Student's $t$-distribution}
We can use the approach from Section~5.2 of \citet{chen2022robust} which uses an approximation by Student's $t$-distribution to recover a confidence level for $\theta$.

\subsection{Permutation-based inference}
Alternatively, we can test the sharp null hypothesis of zero treatment effects across all geos. Under this null, the treatment assignment does not matter and the estimate $\hat{\theta}$ is randomly drawn from the distribution $\{\hat{\theta}_a\}_{a\in\mathcal{A}}$, where $\mathcal{A}$ is the set of all possible treatment assignments. This distribution can be approximated---under the null hypothesis---by repeatedly re-drawing different treatment assignments and computing the corresponding $\hat{\theta}_a$. The null hypothesis is then rejected if the original estimate, $\hat{\theta}$, falls beyond some quantiles (e.g.~10\% or 90\%) of the constructed distribution. 

These tests can also be inverted to produce a confidence interval, but that would require additional assumptions. For instance, for each candidate value $\theta^{*}$ and a null hypothesis $H_0\colon\theta=\theta^{*}$ we could first remove the effect of the treatment from the responses of the treated units (assuming that the null hypothesis holds and $\theta^{*}$ is indeed the true value of $\theta$). Then, we would re-draw the treatment assignment, $a\in\mathcal{A}$, inject the effects (still assuming $\theta=\theta^{*}$) into the responses of the newly treated units, and re-estimate $\hat{\theta}_a$. We would repeat this multiple times to construct the distribution and test the null hypothesis in the same way as we did for $\theta^{*}=0$. We would construct a confidence interval for $\theta$ as the set of all $\theta^{*}$ for which the corresponding $H_0$'s are not rejected at the desired level of statistical significance.
\section{MIP heuristics}\label{sec:heuristics}
For larger $N$ (e.g.~210 US DMAs), it could take weeks
to directly solve the covering MIP formulated in Section~\ref{sec:supergeo_algo}. To accelerate computation, we propose two heuristics that find an approximately optimal solution. Both heuristics reduce the dimensions of variable $x$ by reducing the size of the search space (candidate supergeo pairs).

\paragraph{Partition heuristic.}
The first heuristic that we call \emph{partition heuristic} randomly divides all geos into several partitions, and only includes subsets (supergeo pairs) $G$'s that use geos in the same partition. The number of partitions is a tunable parameter. This way we effectively reduce the number of variables. The advantage of this heuristic is that it is a randomized method, and particularly suitable for running multiple instances in parallel. The disadvantage is that some of the subsets---those that include geos from different partitions---are no longer considered.

\paragraph{Per-geo heuristic}
The second heuristic which we call \emph{per-geo heuristic} sorts the geos according to the magnitude of the pre-test response (which we use as an approximation of the uninfluenced response), and only considers the well-matched supergeo pairs that include the largest geos in the MIP---those geos are usually the hardest to match in the matched pairs design while the smaller geos do not necessarily benefit from supergeo design as much.

Specifically, for each of the $\beta$ geos with largest responses, among all subsets $G$ that contain this geo, we consider the $\alpha$ fraction of subsets that have the smallest score. Here $\beta$ and $\alpha$ are chosen to balance statistical performance with computational efficiency.
The advantage of this heuristic is that it matches the largest geos well, and those are usually the geos that are responsible for the largest errors. This heuristic is also useful for increasing the resulting number of pairs in the design since it preserves the one-to-one pairs of the small geos that are already matched well (see, for example, the supergeo design shown in Figure~\ref{fig:design}). The disadvantage is that we do not consider supergeo pairs of sizes exceeding 2 consisting entirely of the smaller geos.

In practice, these heuristics are not mutual exclusive. We usually run multiple designs---with different hyper-parameters and different random seeds---in parallel and pick the best design among them. Using these heuristics, we can approximately solve the MIP for $N = 210$ and maximum group size $u = 4$ within a few hours on a single machine. 

\section{Proof of NP hardness}\label{sec:np_hard}
In this section we prove Theorem~\ref{thm:np_hard}.
\begin{theorem}[Supergeo design problem is NP-hard. Restatement of Theorem~\ref{thm:np_hard}]
The following problem is NP-hard: Given a set $\mathcal{G}$ of size $|\mathcal{G}|=3m$ and values $Z_g \in \mathbb{Z}^+$ for each $g \in \mathcal{G}$, for any subset $G \subseteq [3m]$ define
\[
\mathrm{score}(G) := \min_{\substack{G_+ \cup G_- = G \\ G_+ \cap G_- = \emptyset}} \Big( \sum_{i \in G_+} Z_{i} - \sum_{j \in G_-} Z_{j} \Big)^2.
\]
Determine whether $\mathcal{G}$ can be partitioned into $m$ disjoint sets $G_1, G_2, \dots, G_m$ such that for all $i \in [m]$, $|G_i|=3$ and $\mathrm{score}(G_i) = 0$ ($[m]$ stands for the set $\{1,2,\dots,m\}$).
\end{theorem}
We prove the NP-hardness of our problem by a reduction from the following numerical 3-dimensional matching problem of \citet{gj90}.
\begin{theorem}[Numerical 3-dimensional matching problem is NP-hard, SP16 in Appendix A of \citet{gj90}]
The following problem is NP-hard: Given disjoint sets $W$, $X$, and $Y$, each containing $m$ elements, each element $a \in W \cup X \cup Y$ having size $s(a) \in \mathbb{Z}^+$, and given a bound $B \in \mathbb{Z}^+$. Determine whether $W \cup X \cup Y$ can be partitioned into $m$ disjoint sets $A_1, A_2, \dots, A_m$ such that for all $i \in [m]$, $A_i$ contains exactly one element from each of $W$, $X$, and $Y$, and $\sum_{a \in A_i} s(a) = B$.
\end{theorem}
\begin{proof}[Proof of Theorem~\ref{thm:np_hard}]
Given a numerical 3-dimensional matching instance, we transform it into a supergeo design instance as follows: Define $M := 1 + B + \sum_{a \in W \cup X \cup Y} s(a)$. Let $\mathcal{G} = W \cup X \cup Y$, and define the values as:
\begin{align*}
Z_w = &~ s(w) + M, \quad \forall w \in W, \\
Z_x = &~ s(x) + 3 M, \quad \forall x \in X, \\
Z_y = &~ B - s(y) + 4 M, \quad \forall y \in Y.
\end{align*}
We prove that $W \cup X \cup Y$ can be partitioned into $m$ disjoint sets with equal sums if and only if $\mathcal{G}$ can be partitioned into $m$ disjoint sets with zero scores.

On the one hand, if $W \cup X \cup Y$ can be partitioned into $m$ disjoint sets $A_1, A_2, \dots, A_m$ such that for all $i \in [m]$, $A_i = \{w_i, x_i, y_i\}$ where $w_i \in W$, $x_i \in X$, and $y_i \in Y$, and $s(w_i) + s(x_i) + s(y_i) = B$, then for all $i \in [m]$ let $G_i = A_i$, and we have
\begin{align*}
\mathrm{score}(G_i) = &~ (Z_{w_i} + Z_{x_i} - Z_{y_i})^2 \\
= &~ \Big((s(w_i) + M) + (s(x_i) + 3M) - (B - s(y_i) + 4M) \Big)^2 \\
= &~ \Big(s(w_i) + s(x_i) - B + s(y_i) \Big)^2 = 0.
\end{align*}

\noindent On the other hand, if $\mathcal{G}$ can be partitioned into $m$ disjoint sets $G_1, G_2, \dots, G_m$ such that for all $i \in [m]$, $|G_i|=3$ and $\mathrm{score}(G_i) = 0$, then we argue that each $G_i$ must be of the form $G_i = \{w_i, x_i, y_i\}$ where $w_i \in W$, $x_i \in X$, $y_i \in Y$, and $s(w_i) + s(x_i) + s(y_i) = B$. Denote $G_i = \{a_i, b_i, c_i\}$, and w.l.o.g.~assume that $\mathrm{score}(G_i) = (Z_{a_i} + Z_{b_i} - Z_{c_i})^2$. Note that the only way that three numbers $a, b, c \in \{1, 3, 4\}$ satisfy $a + b - c = 0$ is when $a = 1$, $b = 3$ (or $a = 3$, $b = 1$), and $c = 4$. This implies that unless $a_i \in W$, $b_i \in X$ (or $a_i \in X$, $b_i \in W$), and $c_i \in Y$, we must have that 
\[
Z_{a_i} + Z_{b_i} - Z_{c_i} = t \cdot M + s,
\]
where $t$ is some \emph{non-zero} integer and $|s| \leq B + s(a_i) + s(b_i) + s(c_i) < 1 + B + \sum_{a \in W \cup X \cup Y} s(a) = M$, and therefore $t \cdot M + s \neq 0$. Thus, $\mathrm{score}(G_i)$ can only possibly be zero when $a_i \in W$, $b_i \in X$ (or $a_i \in X$, $b_i \in W$), and $c_i \in Y$ in which case
\[
s(a_i) + s(b_i) + s(c_i) - B = Z_{a_i} + Z_{b_i} - Z_{c_i} = 0. \qedhere
\]
\end{proof}

\noindent This proof also shows that the supergeo design problem is strongly NP-hard since the numerical 3-dimentional matching problem is strongly NP-hard and in our proof the values $Z_g$'s of the constructed supergeo design instance are all polynomially bounded.


\end{document}